\documentclass{article}

\usepackage{amsmath,amssymb,amsfonts, amsthm}
\usepackage{algorithmic}
\usepackage{graphicx}
\usepackage{textcomp}
\usepackage{xcolor}
\usepackage{bbm}
\usepackage{subcaption}
\usepackage{mwe}
\usepackage{url}
\usepackage{multirow}
\usepackage{float}
\usepackage{eqnarray}
\usepackage{nicefrac}
\usepackage{dsfont}

\usepackage[round]{natbib}
\newtheorem{theorem}{Theorem}[section]
\newtheorem{corollary}{Corollary}[theorem]
\newtheorem{lemma}[theorem]{Lemma}
\newtheorem{definition}{Definition}[section]
\newtheorem*{remark}{Remark}

\newcommand{\raisedchi}{\raisebox{\depth}{\(\chi\)}}

\newcommand{\ie}{{\em i.e.,} }
\newcommand{\eg}{{\em e.g.,} }
\newcommand{\half}{\nicefrac{1}{2}}

\newcommand\pf{{\operatorname{p-f}}}

\newcommand\numberthis{\addtocounter{equation}{1}\tag{\theequation}}

\allowdisplaybreaks

\begin{document}
	
\title{Differential Entropy Rate Characterisations of Long Range Dependent Processes}

\author{Andrew~Feutrill,
	and Matthew~Roughan}


\maketitle

\begin{abstract}
	 A quantity of interest to characterise continuous-valued stochastic processes is the differential entropy rate.
	The rate of convergence of many properties of long range dependent (LRD) processes is slower than might be expected, based on the intuition for conventional processes, \emph{e.g. Markov processes}. Is this also true of the entropy rate?
	
	In this paper we consider the properties of the differential entropy rate of stochastic processes that have an autocorrelation function that decays as a power law. We show that power law decaying processes with similar autocorrelation and spectral density functions, Fractional Gaussian Noise and ARFIMA(0,d,0), have different entropic properties, particularly for negatively correlated parameterisations. Then we provide an equivalence between the mutual information between past and future and the differential excess entropy for stationary Gaussian processes, showing the finiteness of this quantity is the boundary between long and short range dependence. Finally, we analyse the convergence of the conditional entropy to the differential entropy rate and show that for short range dependence that the rate of convergence is of the order $O(n^{-1})$, but it is slower for long range dependent processes and depends on the Hurst parameter.
\end{abstract}


\section{Introduction}

The entropy rate of discrete time stochastic
processes has been studied as a measure of the average
uncertainty. Most investigations of this type have focussed on
processes whose correlations decay quickly, and hence the dependence
on past observations disappears rapidly. However, many real processes
from a variety of contexts, \ie data
networks (\eg \cite{Leland:1993:SNE:166237.166255,
  willinger_self_similar_high_variability_1997, willinger1995}),
climate (\eg \cite{varotsos2006}), hydrology (\eg \cite{beran1994statistics,
  lawrence1977}), and economics (\eg \cite{willinger_stock_1999, cont2005}), have
been shown to exhibit long range dependence, meaning correlations
exist between past and future observations that cannot be ignored at
any time lag. 

Information and coding theory have had profoundly important uses in
signal processing and communication. Noise processes are an important
part of this story. However, in most works, for instance on designing
optimal codes on noisy channels, the noise processes are presumed to
be short range dependent. However, as far back as 1965, ~\cite{1089090}
showed that some noise processes are also long-range
dependent (though the terminology was still developing). 

Recent work has investigated an information theoretic characterisation
of long range and short range processes (e.g. \cite{Li_2004, Chavez_2016,
  ding2016entropic}), using the finiteness of mutual information
between past and future. We aim to clarify this characterisation and
investigate its implications.

This paper calculates the differential entropy rate for the two most
common stationary Gaussian Long Range Dependent (LRD) processes:
Fractional Gaussian Noise (FGN) and the Auto-Regressive
Fractionally-Integrated Moving Average (ARFIMA) process. We start by
deriving the entropy rate for these processes, and show that they both
have negative poles as the processes tend towards strong long-range
correlations, but that their behaviour when anti-correlated is
surprisingly different: FGN has a pole similar to that for positive
correlations, but ARFIMA does not. This contradicts common intuition
based on their similar spectral densities that FGN and ARFIMA(0,d,0) are close to equivalent.

We also investigate the links between the two information measures:
excess entropy and the mutual information between past and future
processes, and compare these to the differential entropy rate. We show
that the differential entropy rate definition for excess entropy is
equivalent to the mutual information between past and future for
continuous valued discrete time Gaussian processes, and hence that
excess entropy is infinite for all long range dependent Gaussian
processes.

Finally, estimators, such as the sample mean, applied to LRD processes
have been shown to have slow convergence rates, which can lead to a
larger than expected uncertainty when investigating these
processes. We ask, \textit{``Does this behaviour apply to estimators
  of entropy?"}  and \textit{``What is the impact of the degree of
  positive or negative correlations on the entropy rate?"}  We show
that while the convergence rate of the conditional entropy of short
range dependent Gaussian processes is in the order of $n^{-1}$, the
rate of convergence for LRD Gaussian processes is slower. Although
this parallels many of the other results for LRD processes, the actual
rate of convergence is different at $O\Big((1-2H)^2 \log(n)/n \Big)$, where $H$ is
the Hurst parameter. As $H \rightarrow 1$ we can see that this
convergence is at its worst, though $\mathcal{H}$ only appears as a factor not in the exponent.

\clearpage

\section{Background}

\subsection{Long range dependence}

LRD refers to a process where correlations decay slower over time such
that the future is non-trivially dependent on the past no matter how
far forward we proceed. A sample path of an LRD fractional Gaussian noise
process, with $H=0.8$, is shown in Figure~\ref{fig: fgn_sample_path}.
As is typical for LRD processes, there is the
appearance of long periods of upwards and downwards ``trends", even
though the process is stationary.

LRD can be defined in two equivalent ways, via the autocorrelation or
spectral density. These definitions are related through the Fourier
Transform~\cite[pg. 117]{brockwell_davis_1986}, and are equivalent via
the Kolmogorov Isomorphism Theorem (\cite{bingham2012}). The following
statement defines the concept of long range dependence in terms of its
autocorrelation function. For reference, the autocovariance function,
$\gamma(k) = \mathds{E}[\left(X_{n+k} - \mu\right)\left(X_n - \mu\right)]$ for a process with mean $\mu$, and autocorrelation coefficient for $k$ lags, $\rho(k)$, are related by
$\rho(k) = \gamma(k)/ \gamma(0) = \gamma(k)/\sigma^2$, where
$\sigma^2$ is the variance of the stochastic process. Note we are using the autocorrelation coefficient and the the autocorrelation function, $\mathds{E}[X_{n+k}X_n]$, this is consistent with many works in the analysis of LRD processes.

\begin{definition}
	Let $\{X_n\}_{n \in \mathbb{N}}$ be a stationary process. If
        there exists $\alpha \in (0,1)$, and $c_\gamma >
        0$, such that the auto-covariance $\gamma(k)$ satisfies
	\begin{align*}
	\lim\limits_{k \rightarrow \infty} \frac{\gamma(k)}{c_\gamma k ^{-\alpha}} = 1,
	\end{align*}
	then we say that the process is long range dependent.
        \label{def:lrd_alpha}
\end{definition}

\noindent The equivalent definition in the frequency domain considers
the limit of the spectral density near the origin.

\begin{definition}
	Let $\{X_n\}_{n \in \mathbb{N}}$ be a stationary process. If
        there exists $\beta \in (0,1)$, and $c_f > 0$,
        such that the spectral density $f(\lambda)$ satisfies
	\begin{align*}
	\lim\limits_{\lambda \rightarrow 0} \frac{f(\lambda)}{c_f |\lambda| ^{-\beta}} = 1,
	\end{align*}
	then we say that the process is long range dependent.
\end{definition}

The concept of LRD is often characterised by the Hurst parameter, $H$, which measures the strength of the correlations between the past and present of a stochastic process. It is
perhaps unfortunate that $H$ is used as standard notation both for the
Hurst parameter, and for Shannon entropy, but we shall side-step that
issue here as we are mainly concerned with entropy rates, which we
will designate with lower-case $h$.

\cite{hurst1951} developed the parameter when he was measuring flows in
the Nile River.  There is a
relationship between $H$ and the $\alpha$ in Definition~\ref{def:lrd_alpha},
namely
\[ H = 1 - \alpha/2. \]
LRD processes have $\alpha \in (0,1)$ and hence have $H \in (0.5,
1)$. The parameter
takes on values between 0 and 1, with $H > \half$ representing a
positively correlated process and $H < \half$ representing a
negatively correlated process and $H = \half$ being a short range
correlated process, such as white Gaussian noise. 

\begin{figure}[t]
	\centering
	\includegraphics[width=\linewidth]{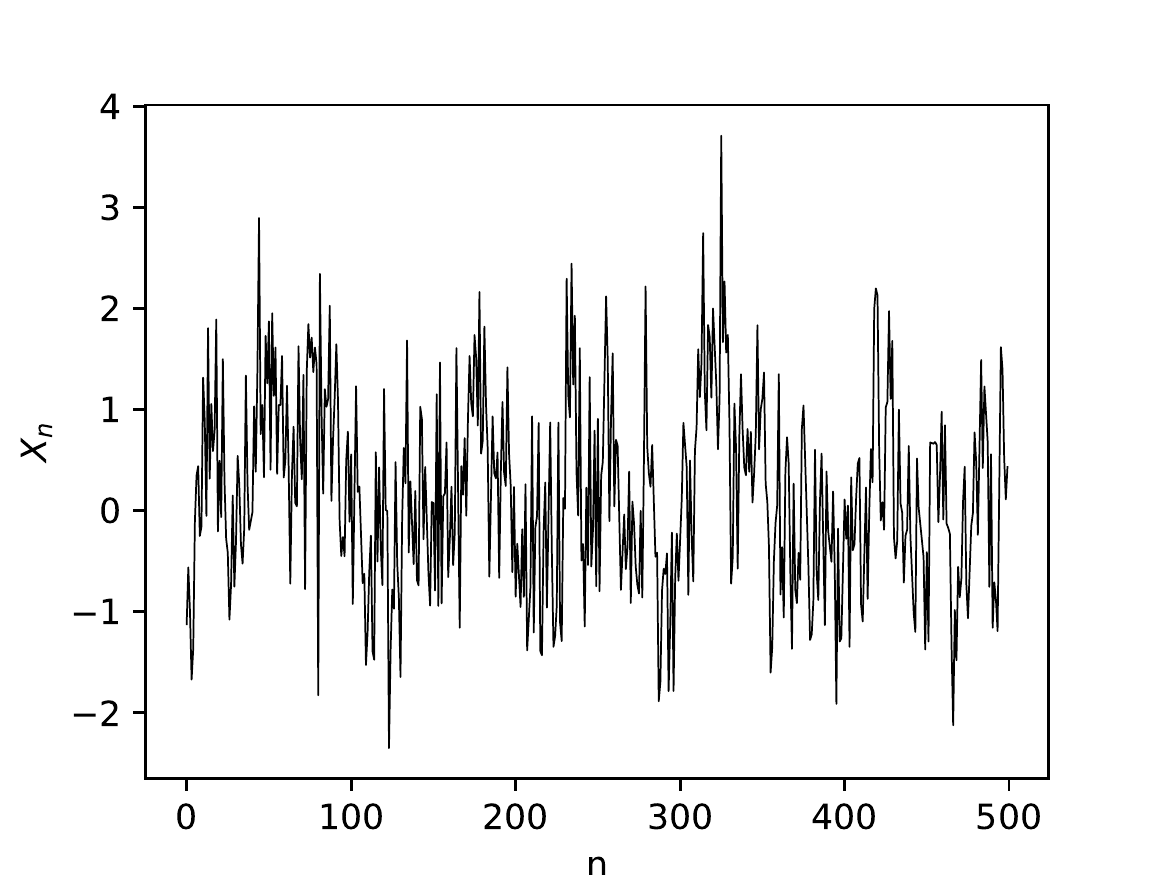}
	\caption{Sample path of Fractional Gaussian Noise with Hurst parameter, $H=0.8$. We can see that the high correlations lead to the appearance of longer trends than we would expect for an independent and identically distributed process.}
	\label{fig: fgn_sample_path}
\end{figure}

Another property of note, that has been used as the definition of LRD
processes, is the (un)summability of the autocorrelation function. For
a LRD processes the sum of the autocovariances diverges, \ie
$\sum_{k=1}^{\infty} \gamma(k) \rightarrow \infty$, whereas for SRD
processes this is finite (\cite{beran1994statistics}). Then we can
interpret LRD as having such strong correlations that the
autocovariance values decay such that the distant past still
influences the future.  The intuition for non-LRD processes is
that the autocovariance function decays exponentially, or quicker, and
in many cases this means, for analysis, we can ignore correlations beyond some short
lag. 

The negatively correlated processes, $H < \half$, have not received as
much consideration at the SRD and LRD cases but we include them in the
analysis here.  These are similar to short range dependent processes,
such as white Gaussian noise, as they still have a summable
autocorrelation function (see~\cite{beran1994statistics,
  gefferth2003nature}). In fact their structure enforces that
$\sum_{k=1}^{\infty} \gamma(k) = 0$. This is quite a strict and
surprising property, and hence~\cite{gefferth2003nature} called these processes
constrained short range dependent (CSRD).

We will be working with Gaussian processes in this paper.
\begin{definition}
	A stochastic process is a Gaussian process if and only if each
        finite collection of random variables from the process has a
        multivariate Gaussian distribution.
\end{definition}
This definition applies to both discrete and continuous time
processes, though here we are principally interested in the
former. Gaussian processes are completely characterised by their
second order statistics (the mean and autocovariance
function) (\cite{box2015time}), which makes them the primary type of
stochastic model used in this context.



\subsection{Entropy rate}

As we are considering continuous random variables in this paper, we
will be considering differential entropy, which is a continuous
extension of Shannon entropy for discrete random variables. In this
paper, we will be using the natural logarithm in all of the
definitions, and hence the units of entropy that we will be working
with are nats. We include standard definitions in order that all
notation be precisely defined.

\begin{definition}
	The differential entropy, $h(X)$, of a continuous random variable with probability density function, $f(x)$, is defined as,
	\begin{align*}
	h(X) = - \int_{\Omega} f(x) \log f(x) \, dx,
	\end{align*}
	where $\Omega$ is the support of the random variable.
\end{definition}

Differential entropy has some important properties which are different
from Shannon entropy. For example, differential entropy can be
negative, or even diverge to $-\infty$, which we can see by
considering the Dirac delta function, $\delta(x)$. The delta, \ie the
unit impulse, is defined by the properties $\delta(x) =0$ for $x \neq
0$ and $\int_{-\infty}^{\infty} \delta(x) \, dx = 1$. The Dirac delta
can be thought of in terms of probability as a completely determined
point in time, that is, a function possessing no uncertainty. It can
be constructed as the limit of rectangular pulses of constant area 1
as their width decreases, and hence we can calculate the entropy of
the delta as
\begin{align*}
	h(X) &= -\int_{-a}^{a} \frac{1}{2a} \log\left(\frac{1}{2a}\right) dx,\\
	     &= \log(2a),
\end{align*}
which tends to $-\infty$ as  $a \rightarrow 0$.

The intuition for $h(X) = -\infty$ from~\cite{cover_thomas_2006} is that the number of bits (note we are working in nats) on
average required to fix a random variable, $X$ to n-bit accuracy
is $h(X) + n$. Meaning $h(X) = -\infty$, can be read as requiring
$n-\infty$ bits, or that we can describe the random variable
arbitrarily accurately without using any bits.
   
We will see the same type of asymptotic behaviour for LRD processes as
$H \rightarrow 1$. Effectively, in the limit the correlations in the
process straight-jacket it, such that the future is completely
determined by the past, and so the incremental uncertainty in the
process is the same as that of the delta.


 

The differential entropy can be extended into the multivariate case
and hence to stochastic processes using the joint entropy for a
collection of random variables.

\begin{definition}
	The joint differential entropy of a collection of random variables $X_1, X_2, ..., X_n$, with density $f(x_1, x_2, ... , x_n) = f(\mathbf{x})$ is defined as,
	\begin{align*}
	h(X_1, X_2, ... , X_n) = - \int_{\Omega}f(\mathbf{x}) \log f(\mathbf{x}) \, d\mathbf{x},
	\end{align*}
	where $\Omega = \Omega_1 \times \Omega_2 \times ... \times \Omega_n$ is the support of the random variables.
\end{definition}

Similarly, the conditional differential entropy can be defined for a random variable, given knowledge of other variables.

\begin{definition}
	The conditional differential entropy of a random variable, $X_n$, given a collection of random variables $X_1, X_2, ... , X_{n-1}$, with a joint density $f(x_1, x_2, ... , x_n) = f(\mathbf{x})$, is defined as,
	\begin{align*}
	h(X_n | X_{n-1}, ... , X_1) = - \int_{\Omega} f(\mathbf{x}) \log f(x_n | x_{n-1}, ... ,x_1) d\mathbf{x},
	\end{align*}
	where $\Omega = \Omega_1 \times \Omega_2 \times ... \times \Omega_n$ is the support of the random variables.
\end{definition}

Finally, we define the concept of differential entropy rate, which can
be thought of as the average amount of new information from each
sample of a random variable in a discrete time process.

\begin{definition}
	Where the limit exists, the differential entropy rate of a
        stochastic process $\raisedchi = \{X_i\}_{i \in \mathbb{N}}$
        is defined to be, 
	\begin{align*}
	h(\raisedchi) &= \lim\limits_{n \rightarrow \infty} \frac{h(X_1, ... , X_n)}{n}.
	\end{align*}
\end{definition}
An example of a process which is non-stationary but has a differential
entropy rate is the Gaussian walk, $S_n$. This is defined as the
process of sums of i.i.d. normally distributed random variables, \ie
$S_n = \sum_{i=1}^{n} X_i$, where $X_i \sim \mathcal{N}(0, \sigma^2)$.
The process has mean 0 for all $n$, however it is non-stationary as
the variance depends on $n$ as,
\begin{align*}
\mbox{Var}(S_n) = \mbox{Var}\left(\sum_{i=1}^{n} X_i\right) = n\sigma^2.
\end{align*}
However, the entropy rate converges and is equal to 
\begin{align*}
\lim\limits_{n \rightarrow \infty} \frac{h(S_1, ... , S_n)}{n} &=  \lim\limits_{n \rightarrow \infty} \frac{\sum_{i=1}^{n} h(S_i | S_{i-1}, \ldots, S_1)}{n},\\
&= \lim\limits_{n \rightarrow \infty}\frac{nh(X_i)}{n}, \\
 &= h(X_i),
\end{align*} 
as each random variable $X_i$ is independent.

An alternative characterisation of the differential entropy rate is given by the following theorem for stationary processes. This was developed for the Shannon entropy of discrete processes~\cite[Theorem 4.2.1]{cover_thomas_2006}, however this has been extended to differential entropy~\cite[pg. 416]{cover_thomas_2006}.

\begin{theorem}[Theorem 4.2.1~{\cite{cover_thomas_2006}}]\label{conditional_entropy_rate}
	For a stationary stochastic process, $\raisedchi = \{X_i\}_{i \in \mathbb{N}}$, the entropy rate is equal to,
	\begin{align*}
	h(\raisedchi) &= \lim\limits_{n \rightarrow \infty} h(X_n | X_{n-1}, ... , X_1).
	\end{align*}
\end{theorem}

\noindent The second equivalent definition is useful because it will
allow us to analyse the convergence rates of conditional entropy to
the differential entropy rate, which is important in estimation of
entropy rates.

\section{Entropy rate function for Fractional Gaussian Noise}\label{entropy_fgn}

We want to understand the effect of memory on the entropic properties
of a stochastic process. We start with the entropy rate
characterisation for Gaussian processes originally derived by
Kolmogorov (\cite[pg. 76]{ihara1993information})
\begin{equation}
  \label{eqn:ihara}
  h(\raisedchi)
  = \frac{1}{2} \log(2\pi e) + \frac{1}{4 \pi} \int_{-\pi}^{\pi} \log (2\pi f(\lambda)) \, d\lambda,
\end{equation}
where $f(\lambda)$ is the spectral density, \ie  the Fourier transform
of the autocovariance function for a mean zero process.

We'll begin by investigating the spectral density of Fractional
Gaussian Noise (FGN), which is given by~\cite{beran1994statistics} as
\begin{equation}
  \label{eqn:fgn_spectral_density}
  f(\lambda) = 2 c_f (1 - \cos\lambda) \sum_{j=-\infty}^{\infty} |2 \pi j + \lambda |^{-2H -1},
\end{equation}
where, $c_f = \frac{\sigma^2}{2 \pi} \sin(\pi H) \Gamma(2H + 1)$, $H$
is the Hurst parameter, and $\sigma^2$ is the variance of the process.

This spectral density is difficult to analyse as it has an infinite sum of absolute values. In particular, when we apply the entropy rate characterisation~\ref{eqn:ihara}, as it involves taking a logarithm of a sum, making analytical calculation prohibitively difficult. However, we can still use this expression to derive some properties of the entropy rate of FGN processes.

\subsection{Comparison of approximate and analytical spectral density for entropy rate calculation}

\noindent Substituting the spectral density of FGN
(\ref{eqn:fgn_spectral_density}) into the second term in the entropy
rate expression (\ref{eqn:ihara}) we get
\begin{align*}
  \int_{-\pi}^{\pi} &\log\big( 2 \pi f(\lambda) \big) \, d \lambda \\
  &= 2 \pi \log(4 \pi c_f) + \int_{-\pi}^{\pi} \log( 1-\cos \lambda ) d \lambda  + \int_{-\pi}^{\pi} \log\left(\sum_{j=-\infty}^{\infty} |2 \pi j + \lambda|^{-2H-1}\right) d\lambda,\\
  &=2 \pi \log(4 \pi c_f) - 2 \pi \log 2 + \int_{-\pi}^{\pi} \log\left(\sum_{j=-\infty}^{\infty} |2 \pi j + \lambda|^{-2H-1}\right) d\lambda.
\end{align*}


\noindent The last term is finite for all $H \in (0,1)$, since the singularity that exists when $\lambda = j = 0$ in the absolute value term is integrable. This is
important as we can then see that this does not affect the asymptotic
behaviour of FGN processes. The resulting entropy rate is
\begin{align*}
h(\raisedchi) =  \frac{1}{2} &\log(2\pi e) + \frac{1}{2}\log\left(2\sigma^2\sin(\pi H) \Gamma(2 H + 1)\right)\\
 &+ \frac{1}{4\pi} \int_{-\pi}^{\pi} \log\left(\sum_{j=-\infty}^{\infty} |2 \pi j + \lambda|^{-2H - 1}\right) d \lambda.
\end{align*}


We calculate $h(\raisedchi)$ using numerical integration via Python's Scipy library (\cite{SciPy-NMeth2020}). We plot the differential entropy rate of Fractional Gaussian Noise as a function of the Hurst parameter, H, in Figure~\ref{fig: fgn_entropy_rate}. The plot shows the impact of the variance on entropy rate calculation, and hence that the entropy rate of Fractional Gaussian Noise has a large dependence on the variance, \emph{i.e.} the second order properties. Each unit increase in variance has a smaller effect on the value of the differential entropy, due to the $\log(\sigma^2)$ term. 

The spectral density expression is quite cumbersome to work with and an approximation is often used, which is accurate at low frequencies~\cite[pg. 53]{beran1994statistics}. It is derived from a Taylor series expansion of the spectral density and is given by,
\begin{align*}
	f(\lambda) \approx c_f |\lambda|^{1-2H}.
\end{align*} 

\noindent We can obtain a closed form for the entropy rate if we substitute this
approximation into the integral in the entropy
rate expression (\ref{eqn:ihara}) to get
\begin{align*}
	\int_{-\pi}^{\pi} \log (2\pi f(\lambda)) \, d\lambda
	 &=\int_{-\pi}^{\pi} \log\left(2\pi c_f\right) d\lambda +  \int_{-\pi}^{\pi} \log\left(|\lambda|^{1-2H}\right) \, d\lambda, \\
	&= 2\pi \log\left(2\pi c_f\right) + 2\left(1 -2H\right)\int_{0}^{\pi} \log\left(\lambda\right) \, d\lambda,\\
	&= 2\pi \log\left(2\pi c_f\right) + 2\left(1 -2H\right)\left(\pi \log \pi - \pi\right).
\end{align*}

\noindent Note that there is a singularity at the origin of the
spectral density of LRD processes. However, the integral is still well
defined and finite in this case.  Therefore the entropy rate
approximation is,
\begin{align*}
  \tilde{h}(\raisedchi) =  \frac{1}{2}\log(2\pi e) 
  + \frac{1}{2}\log 2 \pi c_f
  + (1-2H)(\log\pi - 1),
\end{align*}
which differs from the exact formulation only in the last term.  

\begin{figure}
	\centering
	\includegraphics[width=\linewidth]{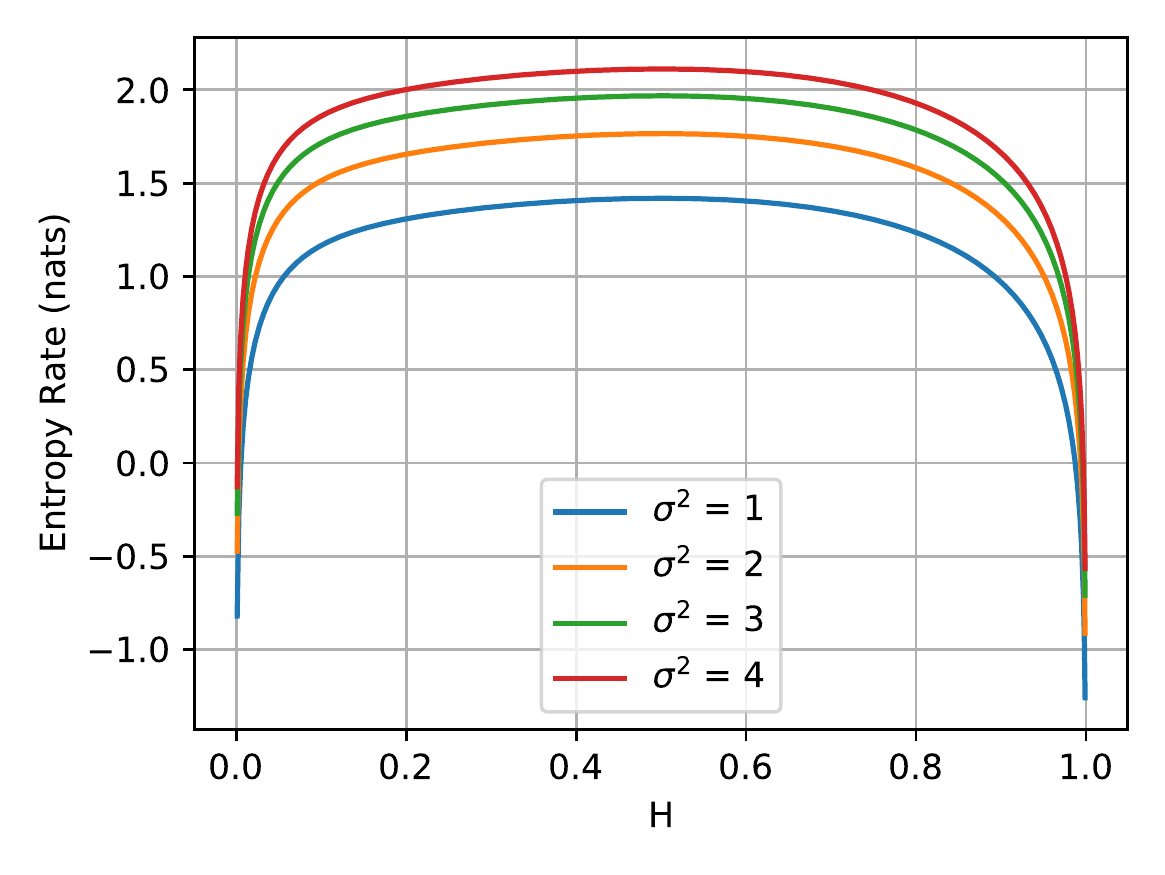}
	\caption{Entropy rate of Fractional Gaussian Noise as a function of the Hurst Parameter. The maximum is at $H=0.5$, where the process is white Gaussian noise. As $H \rightarrow 0$ or $1$, the function tends towards $-\infty$, as the strength of the negative or positive correlations increase. The impact of changing variance decreases as the variance increase, due to the $\log(\sigma^2)$ term.}
	\label{fig: fgn_entropy_rate}
\end{figure}

\begin{figure}
	\centering
	\includegraphics[width=\linewidth]{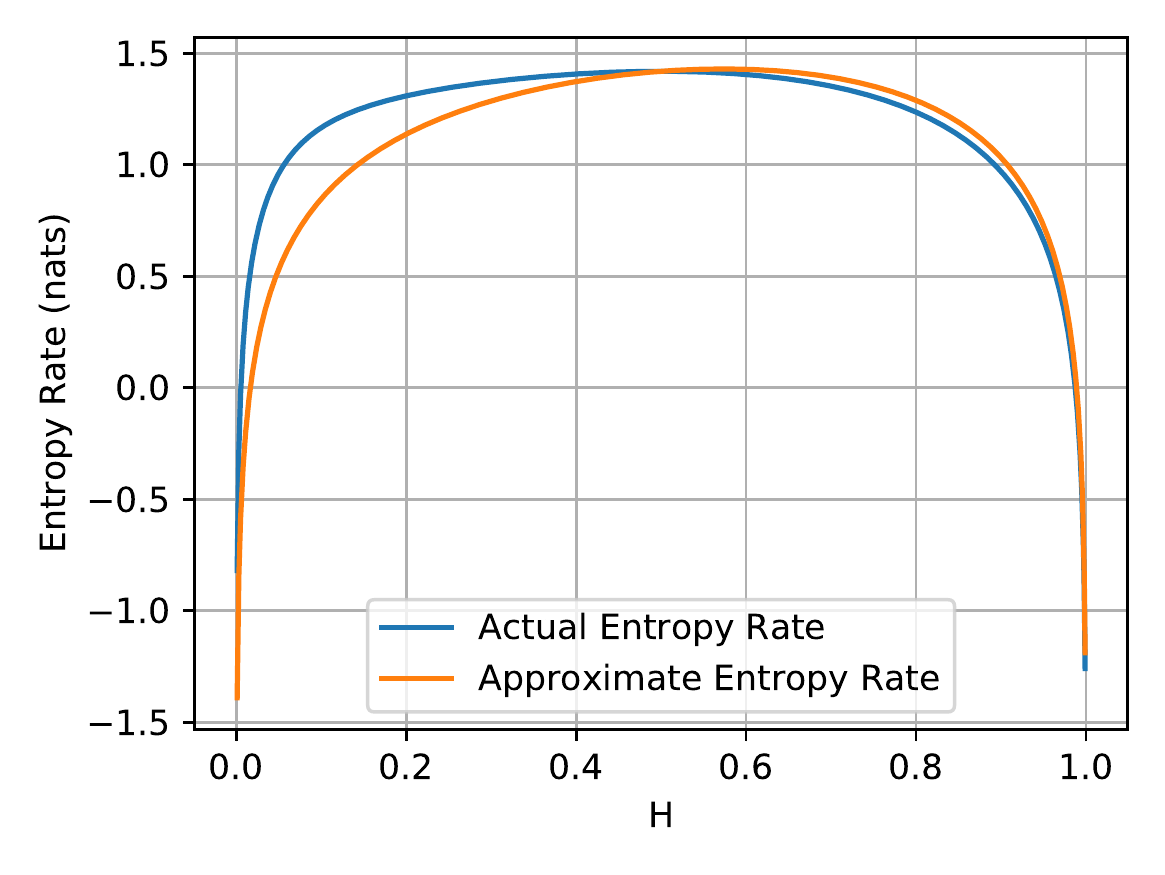}
	\caption{Comparison of the numerically integrated spectral
          density and the spectral density approximation. The
          approximation is relatively good for $H \geq \half$ but an
        underestimate for $H \leq \half$.}
	\label{fig: fgn_entropy_rate_comparison}
\end{figure}


Figure~\ref{fig: fgn_entropy_rate_comparison} shows the
entropy rate and its approximation. We can see that the entropy rate
approximation is very good for the positively correlated cases $H \geq
0.5$ and at the limits around $H=0$ or $1$. However for moderately,
negatively-correlated processes the approximation is a noticeable
underestimate of the entropy rate.

\subsection{Properties of Entropy rate for Fractional Gaussian Noise}

Figure~\ref{fig: fgn_entropy_rate_comparison} shows some interesting
properties
\begin{itemize}
	\item The entropy rate as a function of $H$ is not symmetric. Negatively
          correlated processes seem to have higher uncertainty the
          same distance from $H=0.5$. 
	\item The entropy rate asymptotically tends to $-\infty$ as $H
          \rightarrow$ 0 or 1. 
	\item The maximum entropy rate occurs at 0.5. Indicating that
          the maximum entropy occurs for white Gaussian noise.  
\end{itemize}
We explain how these properties emerge below.

\subsubsection{Asymptotic behaviour}


\begin{theorem}
	The approximate differential entropy rate of Fractional
        Gaussian Noise, $\tilde{h}(\raisedchi) \rightarrow -\infty$ as $H
        \rightarrow 0$ or $1$.
\end{theorem}

 \begin{proof}
 	 When $H \rightarrow 0$ or $1$, the term $c_f \rightarrow 0$,
         as the gamma function terms are finite, however the
         trigonometric terms tend to 0 as $H$ tends to an integer
         value. Hence, asymptotically the approximate entropy rate
         expression is dominated by $\log c_f \rightarrow -\infty$, as
         $c_f \rightarrow 0$.
 \end{proof}

 \begin{remark}
   Note that the approximation works well in the limits $\mathcal{H} \rightarrow
   0$ or 1, and so the theorem describes the asymptotic behaviour of
   entropy rate well. Moreover, the theorem lines up with the
   intuition for an LRD process. As we move closer to either perfectly
   positively or negatively correlated, the process becomes ``less
   uncertain", \emph{i.e.}, we have less entropy on average. When the
   uncertainty disappears, by viewing the entire past we can
   accurately infer the current value. It's important to reiterate
   that the differential entropy can be $-\infty$, which can be
   interpreted as least uncertainty for a process.
\end{remark}

\subsubsection{Maximum}

We want to understand the maximum of differential entropy rate, as a
function of the Hurst parameter. This will provide an understanding of
which parameter choices represent the highest uncertainty. We
differentiate the entropy rate, with respect to $H$ and then solve for
$H$ when the derivative equals zero. Here we need to apply this to the
exact formula because the approximation distorts the location of the
maximum. Therefore, dropping constant terms, we get
\begin{align*}
  \frac{dh}{dH}
  &=  \frac{1}{2} \frac{d}{dH} \log\left(\sigma^2\sin(\pi H) \Gamma(2 H + 1)\right)\\
   & \qquad+  \frac{1}{4 \pi} \frac{d}{dH} \int_{-\pi}^{\pi} \log\left(\sum_{j=-\infty}^{\infty} |2 \pi j + \lambda|^{-2H - 1}\right) d \lambda  \\
  &=  \frac{1}{2} \frac{d}{dH} \log\big(\sin(\pi H)\big) +  \frac{1}{2} \frac{d}{dH} \log\big(\Gamma(2 H + 1)\big)\\
  &\qquad - \frac{1}{2 \pi} \int_{-\pi}^{\pi} \frac{\sum_{j} \log(|2 \pi j + \lambda|) |2 \pi j + \lambda|^{-2H -1}}{\sum_{j} |2 \pi j + \lambda|^{-2H -1}} d\lambda   \\
  &= \frac{\pi}{2} \cot(\pi H) +  \psi(2 H + 1) - \frac{1}{2 \pi} \int_{-\pi}^{\pi} \frac{\sum_{j} \log(|2 \pi j + \lambda|) |2 \pi j + \lambda|^{-2H -1}}{\sum_{j} |2 \pi j + \lambda|^{-2H -1}} d\lambda.
\end{align*}


\noindent where $\psi(z) = \Gamma^\prime(z)/\Gamma(z)$ is the digamma
function.

Then we set this expression to zero, and solve for $\mathcal{H}$. This is a
transcendental equation with no closed form. We solve it numerically
using Python's SciPy package (\cite{SciPy-NMeth2020}), which yields $\mathcal{H}
\approx 0.500$. Therefore we conclude that the maximum entropy rate, using the exact spectral density, is at $\mathcal{H}=0.5$, which aligns with the idea that a SRD process has more uncertainty than any equivalent LRD
process.

Note that from the solution of the spectral density approximation is
$H \approx 0.516$. So although using the spectral density
approximation is acceptable for many purposes, it can lead to false
conclusions about the properties of the differential entropy rate.


\section{Entropy rate function for ARFIMA(\MakeLowercase{p},\textbf{d,q)}}

We consider the differential entropy rate function of a related process to Fractional Gaussian Noise, which is ARFIMA(p,d,q), the fractional extension of the ARIMA (Autoregressive Integrated Moving Average) processes, by extending to non-integer differencing parameters, $d$ (see~\cite{hosking1981, granger1980}). FGN and ARFIMA(0,d,0) are commonly used stationary LRD processes for modelling real phenomena, and in particular FGN and ARFIMA(0,d,0) have very similar properties in the time and frequency domains. Additionally, these processes have been linked by by limit of their autocorrelation coefficients, $\rho(k) := \gamma(k)/\sigma^2$, under aggregation and rescaling (see~\cite{gefferth2003nature}). However, ARFIMA processes do differ from FGN in that you could change the fixed point, \emph{i.e.}, alter the eventual limit under aggregation and rescaling, with the addition of additive noise (\cite{veitch2013farima}), \emph{i.e.}, this class is less robust to the addition of noise. Hence, there may be some differences in behaviour when looking through an entropic lens.

Before we define an ARFIMA(p,d,q) process, we define two polynomials that are required for the ARFIMA definition. These are polynomials of the lag operator, $L$, where $LX_n = X_{n-1}$, defined by
\begin{align*}
	\phi(x) &= 1 - \sum_{j=1}^{p} \phi_j x^j, \text{for coefficients }\phi_j \text{ and } p \in \mathbb{Z^+},\\
	\psi(x) &= 1 + \sum_{j=1}^{q} \psi_j x^j, \text{for coefficients }\psi_j \text{ and } q \in \mathbb{Z^+}.
\end{align*}

\noindent Now we can define an ARFIMA(p,d,q) process,

\begin{definition}
	For a stationary stochastic process $\{X_i\}_{i \in \mathbb{Z}}$, such that
	\begin{align*}
		\phi(L)(1-L)^dX_i &= \psi(L)\epsilon_i ,
	\end{align*}
	for some $-\half < d < \half$ and $\epsilon_i$, a zero mean normally distributed random variable with variance $\sigma_\epsilon^2$. $\{X_i\}_{i \in \mathbb{Z}}$ is called a ARFIMA$(p,d,q)$ process. 
\end{definition}

There is a connection between the differencing parameter, $d$, and the Hurst parameter, $H$, for these processes

An ARFIMA(0,d,0) process is a special case of an ARFIMA(p,d,q), where $\phi(x) = \psi(x) = 1$, that is there is no lag on the noise, $\epsilon$, and all the auto-regressive lags on the previous values come from the differencing operator.

The spectral density of an ARFIMA(p,d,q) is given by~\cite{beran1994statistics} as
\begin{align}
	f(\lambda) &= \frac{\sigma_\epsilon^2|\psi(e^{i\lambda})|^2}{2\pi|\phi(e^{i\lambda})|^2}|1 - e^{i\lambda}|^{-2d}.\label{eqn:arfima_spec_dens}
\end{align}

\noindent The following theorem from \cite{hosking1981} (and in~\cite{beran1994statistics}) gives infinite autoregressive and moving average representations for ARFIMA(0,d,0) processes. 

\begin{theorem}[Proposition 2.2~{\cite{beran1994statistics}}]
	Let $X_n$ be a fractional ARIMA(0,d,0) process with -$\frac{1}{2} < d < \frac{1}{2}$ . Then\\
	(i) the following infinite autoregressive representation holds:\\
	\begin{align*}
		\sum_{k=0}^{\infty} \pi_k X_{n-k} = \epsilon_n,
	\end{align*} 
	where $\epsilon_n (n = 1, 2, ...)$ are independent identically distributed random variables and\\
	\begin{align*}
		\pi_k = \frac{\Gamma(k-d)}{\Gamma(k+1)\Gamma(-d)}.
	\end{align*}
	For $k \rightarrow \infty$ we have,
	\begin{align*}
		\pi_k \sim \frac{1}{\Gamma(-d)}k^{-d-1}.
	\end{align*}
	(ii) The following infinite moving average representation holds:\\
	\begin{align*}
		X_n = \sum_{k=0}^{\infty} a_k \epsilon_{n-k}
	\end{align*} 
	where $\epsilon_n (n = 1, 2, ...)$ are independent identically distributed random variables and\\
	\begin{align*}
		a_k = \frac{\Gamma(k+d)}{\Gamma(k+1)\Gamma(d)}.
	\end{align*}
	For $k \rightarrow \infty$ we have
	\begin{align*}
		a_k \sim \frac{1}{\Gamma(d)}k^{d-1}.
	\end{align*}
\end{theorem}


We will express an entropy rate characterisation for ARMA processes in terms of its innovation process variance, from \cite{ihara1993information}, and show that this can be extended to ARFIMA(0,d,0) and ARFIMA(p,d,q) processes. Then we will use the result to characterise the entropy rate of an ARFIMA(0,d,0) process in terms of its process variance.

\begin{theorem}[Proposition 2.2~{\cite{ihara1993information}}]\label{thm:ihara}
	The entropy rate of an ARMA(p,q) process is given by,
	$h(\raisedchi) = \frac{1}{2}\log(2\pi e \sigma_\epsilon^2)$.
\end{theorem}

\noindent Now, we state our extension to ARFIMA(0,d,0) and present a proof based on Ihara's proof of Theorem~\ref{thm:ihara}

\begin{theorem}\label{arfima(0,d,0)}
	The entropy rate of an ARFIMA(0,d,0) process is given by,
	$h(\raisedchi) = \frac{1}{2}\log(2\pi e \sigma_\epsilon^2)$.
\end{theorem}

\begin{proof}
	First we calculate $\int_{-\pi}^{\pi} \log(2\pi f(\lambda)) d\lambda$, using the spectral density of an ARFIMA process given in~(\ref{eqn:arfima_spec_dens})
	\begin{align*}
	\int_{-\pi}^{\pi} \log (2\pi f(\lambda)) d\lambda &= \int_{-\pi}^{\pi} \log(\sigma_\epsilon^2|1 - e^{i \lambda}|^{1-2H}) d\lambda,\\
	&= \int_{-\pi}^{\pi} \log(\sigma_\epsilon^2) d\lambda + (1-2H)\int_{-\pi}^{\pi} \log|1 - e^{i \lambda}| d\lambda.
	\end{align*}
	Now we transform the elements in the last term using their trigonometric representation, 
	\begin{align*}
	|1 - e^{i\lambda}| &= |1 - \cos(\lambda) - i\sin(\lambda)|,\\
	&= \sqrt{(1 - \cos(\lambda))^2 + \sin^2(\lambda)},\\
	&= \sqrt{2 - 2\cos(\lambda)},\\
	&= \sqrt{4 \sin^2\left(\frac{\lambda}{2}\right)},\\
	&= 2|\sin\left(\frac{\lambda}{2}\right)|.
	\end{align*}
	This makes the integral of the log spectral density,
	\begin{align*}
	\int_{-\pi}^{\pi} \log |1 - e^{i\lambda}| d\lambda &= 2 \int_{0}^{\pi} \log\left(2\sin\left(\frac{\lambda}{2}\right)\right) d\lambda,\\
	&= 2 \int_{0}^{\pi} \log(2) d\lambda + 2 \int_{0}^{\pi} \log\left(\sin\left(\frac{\lambda}{2}\right)\right) d\lambda,
	\end{align*}
	We substitute  $y = \lambda/2$,
	\begin{align*}
	\int_{-\pi}^{\pi} \log |1 - e^{i\lambda}| d\lambda &= 2 \pi \log(2) +  2 \int_{0}^{\frac{\pi}{2}} \log(\sin y) 2 dy,\\
	&= 2 \pi \log(2) + 4\left(-\frac{\pi}{2} \log(2)\right), \\
	&= 0.
	\end{align*}
	Where the equality $\int_{0}^{\frac{\lambda}{2}} \log(\sin y) dy = -\frac{\pi}{2} \log(2)$ is given  by~\cite{koyama2005}.
	
	\noindent So the last term of the spectral density vanishes, and 
	\begin{align*}
	\int_{-\pi}^{\pi} \log (2\pi f(\lambda)) d\lambda &= \int_{-\pi}^{\pi} \log(\sigma_\epsilon^2) d\lambda + (1-2H)\int_{-\pi}^{\pi} \log|1 - e^{i \lambda}| d\lambda,\\
	&= 2\pi \log(\sigma_\epsilon^2).
	\end{align*}
	Using Kolmogorov's entropy rate expression, the entropy rate is therefore,
	\begin{align*}
	h(\raisedchi) &= \frac{1}{2} \log (2 \pi e) +  \frac{1}{4\pi}\left(2\pi \log(\sigma_\epsilon^2)\right),\\
	&= \frac{1}{2} \log(2 \pi e \sigma_\epsilon^2).
	\end{align*}
\end{proof}

\begin{remark}
	This can be shown also using the infinite autoregressive expression above, $X_n = \epsilon_n - \sum_{k=1}^{\infty} \pi_kX_{n-k}$, and substituting into the conditional entropy rate for stationary processes, $h(\raisedchi) = \lim\limits_{n \rightarrow \infty} h(X_n | X_n-1, ... , X_0)$. Then we can remove the conditioning from the entropy rate calculation, \begin{align*}
		h(\raisedchi) = \lim\limits_{n \rightarrow \infty} h\left(\epsilon_n - \sum_{k=1}^{\infty} \pi_kX_{n-k} \bigg| X_n-1, ... , X_0\right) = \lim\limits_{n \rightarrow \infty} h(\epsilon_n).
	\end{align*} Which then implies that $h(\raisedchi) = \frac{1}{2}\log(2\pi e \sigma_\epsilon^2)$, \emph{i.e.}, the entropy rate of the process depends only on the entropy introduced at each step by the innovations. Therefore, we conclude that the entropy rate of an ARFIMA(0,d,0) process depends on the innovation variance and not $H$.
\end{remark}

We can generalise to ARFIMA(p,d,q) processes by adding an additional condition, the invertibility of the moving average polynomial. This is an extremely common condition applied in the theory of autoregressive-moving average, \emph{i.e.}, ARMA(p,q) processes. The condition implies that all roots of the moving average polynomial lie outside of the unit circle, and similarly the stationarity condition of the process ensures that all roots of the autoregressive polynomial lie outside of the unit circle (\cite{box2015time}). We will use these conditions on the ARFIMA processes to analyse their properties.

\begin{theorem}\label{arfima(p,d,q)}
	The entropy rate of a stationary ARFIMA(p,d,q) process with invertible moving average polynomial is given by,
	$h(\raisedchi) = \frac{1}{2}\log(2\pi e \sigma_\epsilon^2)$.
\end{theorem}
\begin{proof}
	Since ARFIMA(p,d,q) processes are stationary and invertible, this implies that the polynomials $\phi(x)$ and $\psi(x)$, have roots outside of the unit circle, i.e. each root $z \in \mathbb{C}$ is such that $|z| > 1$. By the Fundamental Theorem of Algebra, both the autoregressive and moving average polynomials can be factored into linear factors. As the constant terms are 1, this implies the polynomials can be factored as $\phi(x) = \prod_{i=1}^p (1 - a_i e^{i \lambda})$ and $\psi(x) = \prod_{i=i}^{q} (1 - b_i e^{i \lambda})$, where $|a_i|, |b_i| < 1, \forall i$. 
	Recall from equation~(\ref{eqn:arfima_spec_dens})that the spectral density is given by,
	\begin{align*}
		f(\lambda) &= \frac{\sigma_\epsilon^2}{2\pi} |1 - e^{i\lambda}|^{-2d} \frac{|\psi(e^{i\lambda})|^2}{|\phi(e^{i\lambda})|^2}.
	\end{align*}
	\noindent Hence,
	\begin{align*}
		\log((2\pi f(\lambda)))
		&= \log(\sigma_\epsilon^2) -2d\log|1-e^{i\lambda}| \\
		&\qquad + 2\log\left|\prod_{j=1}^{q}(1 - a_j e^{i\lambda})\right| -2\log\left|\prod_{j=1}^{p}(1 - b_j e^{i\lambda})\right|,\\
		&= \log(\sigma_\epsilon^2) -2d\log|1-e^{i\lambda}| \\
		&\qquad + \sum_{j=1}^{q} 2\log|1 - a_j e^{i\lambda}| - \sum_{j=1}^{p} 2\log|1 - b_j e^{i\lambda}|.
	\end{align*}
	Now we calculate the integral of the log spectral density,
	\begin{align*}
		\int_{-\pi}^{\pi} \log(2\pi f(\lambda)) d \lambda
		&= \int_{-\pi}^{\pi} \log(\sigma_\epsilon^2) d \lambda - \int_{-\pi}^{\pi} 2d\log|1-e^{i\lambda}| d \lambda\\
		&\qquad + \sum_{j=1}^{q} 2 \int_{-\pi}^{\pi} \log|1 - a_j e^{i\lambda}| d \lambda - \sum_{j=1}^{p} 2 \int_{-\pi}^{\pi} \log|1 - b_j e^{i\lambda}|  d \lambda,\\
		&= 2 \pi \log(\sigma_\epsilon^2).
	\end{align*}
	Where the third equality is given as all the integrals of $\log|1 - ae^{i\lambda}|$ over $[-\pi, \pi]$ vanish for $|a| \le 1$ (\cite{shiu_2004}).
	
	We substitute this expression into Kolmogorov's entropy rate expression for Gaussian processes.
	\begin{align*}
		h(\raisedchi) 
		&= \frac{1}{2}\log(2 \pi e) + \frac{1}{4 \pi} (2 \pi \log(\sigma_\epsilon^2)),\\
		&= \frac{1}{2}\log(2 \pi e \sigma_\epsilon^2).
	\end{align*}
\end{proof}

This result leads to the following corollary, which can finalise the discussion of the differential entropy rate in terms of innovation variance for the classes of AR, MA, ARMA processes. This is relevant as the definition in terms of the innovation variance is the perspective that is commonly used in the time series literature, when modelling real world processes.

\begin{corollary}
	The differential entropy rate of stationary AR(p), invertible MA(q) and, stationary and invertible ARMA(p,q) processes is $h(\raisedchi) = \frac{1}{2} \log(2 \pi e \sigma_\epsilon^2)$.
\end{corollary}

\noindent Hence, for these models the entropy rate can be calculated in terms of the variance of its innovations. However we want to compare the entropy rates, as a function of their Hurst parameter, between ARFIMA(0,d,0) and FGN, so we want to fix the variance of process itself, $\sigma^2$. We will use the autocovariance function of ARFIMA(0,d,0), from~\cite{beran1994statistics}, 
\begin{align*}
	\gamma(k) &= \sigma_\epsilon^2 \frac{(-1)^k \Gamma(1-2d)}{\Gamma(k-d+1)\Gamma(1-k-d)}.\\
\end{align*}
Note that $\gamma(0) = \sigma^2$,
\begin{align*}
	\sigma^2 = \gamma(0) &= \sigma_\epsilon^2 \frac{ \Gamma(1-2d)}{\Gamma(1-d)^2},
\end{align*}
and hence,
\begin{align*}
	 \sigma_\epsilon^2 &= \sigma^2\frac{\Gamma(1-d)^2}{\Gamma(1-2d)}.
\end{align*}

This leads to the following characterisation of ARFIMA(0,d,0) processes in terms of the Hurst parameter, $H$, noting that $d = H - \half$.

\begin{corollary}
	The entropy rate of an ARFIMA(0,d,0) process for a fixed process variance, $\sigma^2$, is given by, 
	\begin{align}
	h(\raisedchi) = \frac{1}{2} \log(2 \pi e \sigma^2) + \log\left(\Gamma\left(\frac{3}{2} - \mathcal{H}\right)\right) - \frac{1}{2} \log \bigg(\Gamma(2 - 2\mathcal{H})\bigg). \label{eq:arfima_entropy_rate}
	\end{align}
\end{corollary}

\begin{proof}
	By Theorem~\ref{arfima(0,d,0)} and from the characterisation of $\sigma_\epsilon^2$ above,
	\begin{align*}
		h(\raisedchi) &= \frac{1}{2} \log \left(2 \pi e \sigma^2\frac{\Gamma(1-d)^2}{\Gamma(1-2d)}\right),\\
		&= \frac{1}{2} \log(2 \pi e \sigma^2) + \log(\Gamma(1-d)) - \frac{1}{2} \log(\Gamma(1-2d)),\\
		&= \frac{1}{2} \log(2 \pi e \sigma^2) + \log\left(\Gamma\left(\frac{3}{2} - \mathcal{H}\right)\right) - \frac{1}{2} \log \bigg(\Gamma(2 - 2\mathcal{H})\bigg).
	\end{align*}
\end{proof}
\begin{remark}
	The same approach can be used for more general ARFIMA(p,d,q) processes. However in this case, there is no general closed form for the autocovariance function, so the variance must be calculated for each process and then substituted for the innovation variance. Interestingly, this result indicates that the effect of the changing the process variance is balanced by the effect of the change in the Hurst parameter, with respect to the innovation variance. This results in the constant differential entropy rate when considered in terms of its innovation variance.
\end{remark}

\begin{figure}
	\centering
	\includegraphics[width=\linewidth]{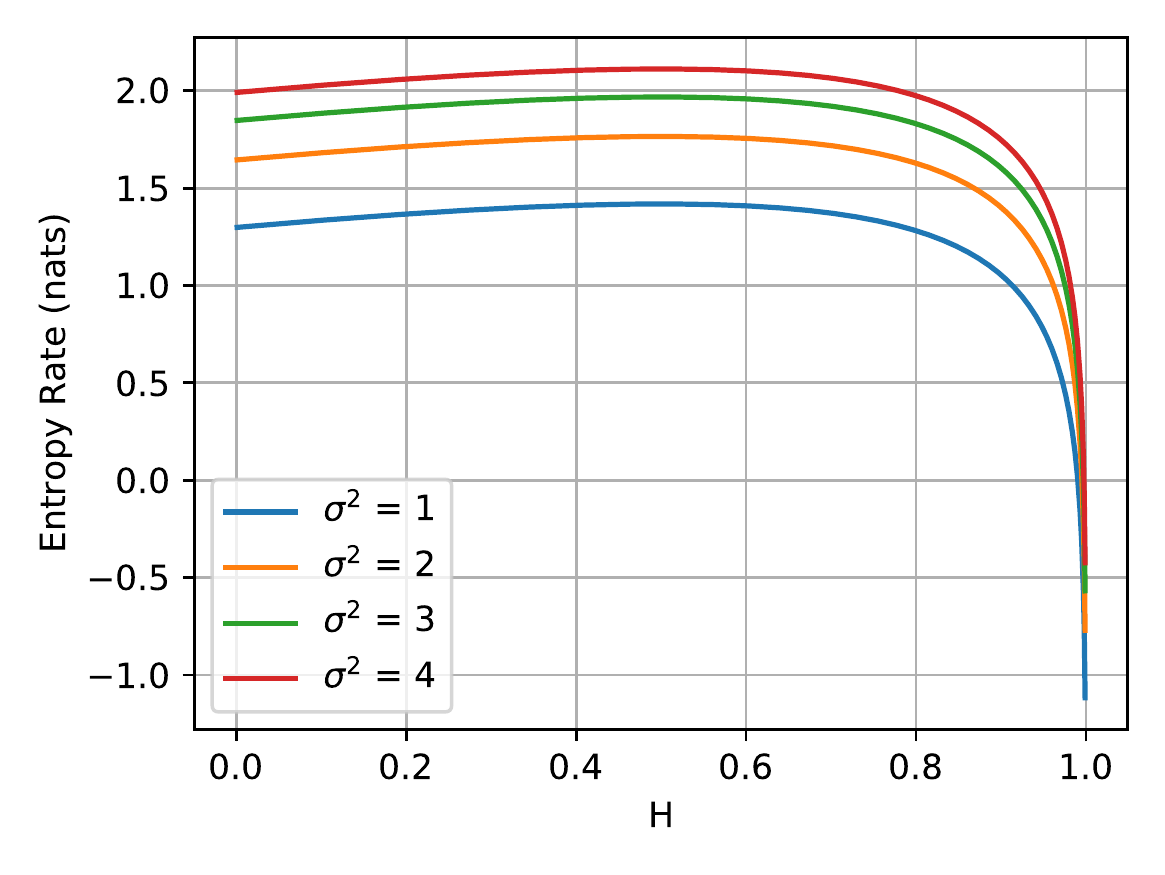}
	\caption{The entropy rate of ARFIMA(0,d,0) as a function of the Hurst parameter, $H$, for variance, $\sigma^2 = 1, 2, 3, 4$. On the positively correlated side, $H > 0.5$, we see a similar asymptotic behaviour to FGN. However, for negatively correlated processes, the amount of entropy in the process, stays quite high. We see the maximum of the function at $H=0.5$, which intuitively shows that the highest uncertainty occurs for the white Gaussian noise process.}
	\label{fig: arfima_entropy_rate_function}
\end{figure}

We show the plot of the ARFIMA(0,d,0) entropy rate as a function of the Hurst parameter, $H$, with process variance, $\sigma^2 = 1,2,3,4$, in Figure~\ref{fig: arfima_entropy_rate_function}. The plot shows some interesting behaviour, particularly when compared to the FGN entropy rate function in Figure~\ref{fig: arfima_fgn_entropy_rate_function_comparison}. Some of these observed properties are:
\begin{itemize}
	\item The entropy rate is not symmetric, much less so than FGN. The positively correlated side has a dramatic drop, however the negatively correlated side stays relatively high. In order words, there is a demonstrable difference between FGN and ARFIMA(0,d,0) in the behaviour as CSRD processes.
	\item The entropy rate asymptotically tends to $-\infty$ as $H \rightarrow$ 1 only. 
	\item The maximum entropy rate occurs at the same point as FGN, $\mathcal{H} =0.5$, indicating that the maximum entropy occurs for white Gaussian noise.
\end{itemize}

\begin{figure}
	\centering
	\includegraphics[width=\linewidth]{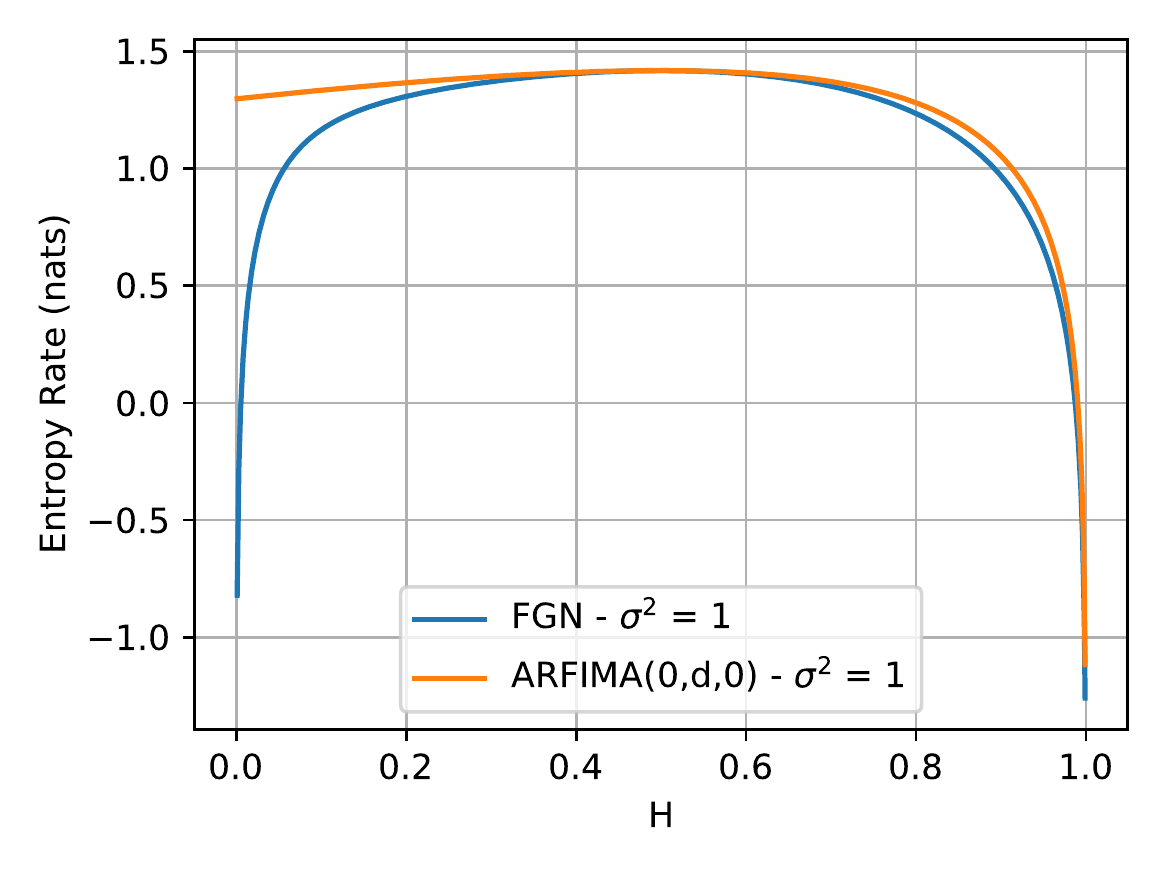}
	\caption{The comparison of the entropy rate as function of the Hurst parameter, for both ARFIMA(0,d,0) and FGN processes, with variance 1. It appears that the ARFIMA(0,d,0) process has an entropy rate which is greater than or equal to FGN for all values of $H$. The negatively correlated portion falls away quickly as $H \rightarrow 0$ for FGN but stays relatively high for the ARFIMA(0,d,0) process. The maximum of the functions coincide at $H=0.5$.}
	\label{fig: arfima_fgn_entropy_rate_function_comparison}
\end{figure}

\noindent Similar to the previous section, we will prove the asymptotics of the entropy rate function, and show that the maximum occurs at $H=0.5$.

\begin{corollary}
	The differential entropy rate of ARFIMA(0,d,0), $h(\raisedchi)$, tends to negative infinity as $H \rightarrow 1$, for a fixed variance $\sigma^2 \in \mathbb{R}$.
\end{corollary}
\begin{proof}
	As $H \rightarrow 1$, the term $\Gamma(\frac{3}{2} - H)$ is finite, as well as the $\frac{1}{2} \log(2 \pi e \sigma^2)$, for a fixed variance $0 < \sigma^2 < \infty$. 
	Now, as $H \rightarrow 1$, the term $\Gamma(2 - 2H) \rightarrow  \Gamma(0)$. There exists a singularity for the gamma function at 0, which diverges to infinity. Which implies that the term $-\frac{1}{2} \log(\Gamma(2 - 2H)) \rightarrow -\infty$, since $\Gamma(x) \rightarrow \infty$, as $x \rightarrow 0$. This implies that $h(\raisedchi) \rightarrow -\infty$, as $H \rightarrow 1$.
\end{proof}
\begin{remark}
	Note that the value of the entropy rate function for an ARFIMA(0,d,0) process as $H \rightarrow 0$, is 
	\begin{align*}
		h(\raisedchi) = \frac{1}{2} \log(2 \pi e \sigma^2) + \log(\Gamma(\frac{3}{2})) - \frac{1}{2} \log(\Gamma(2)) \approx 1.298.
	\end{align*}
\end{remark}

To complete this section of the analysis, we will consider the maximum of the entropy rate function of ARFIMA(0,d,0), and conclude which Hurst parameter has the highest uncertainty, in the sense of maximum differential entropy rate. 

\begin{theorem}
	The differential entropy rate of ARFIMA(0,d,0) as a function of $H$ attains the maximum at $H = \half$.
\end{theorem}

\begin{proof}
We differentiate the entropy rate function with respect to $H$, and we get
\begin{align*}
	\frac{dh(\raisedchi)}{dH} &= \frac{d}{dH}\left(\frac{1}{2} \log(2 \pi e \sigma^2) + \log\left(\Gamma(\frac{3}{2} - H)\right) - \frac{1}{2} \log\left(\Gamma(2 - 2H)\right)\right),\\
	&= \frac{d}{dH}\left(\log(\Gamma(\frac{3}{2} - H))\right) - \frac{d}{dH}\left(\log(\Gamma(2 - 2H))\right),\\
	&= \frac{\Gamma(\frac{3}{2} - H)\psi\left(\frac{3}{2} - H\right)}{\Gamma(\frac{3}{2} - H)} - \frac{\Gamma(2 - 2H)\psi(2 - 2H)}{\Gamma(2 - 2H)},\\
	&= \psi\left(\frac{3}{2} - H\right) - \psi(2 - 2H),
\end{align*}
where $\psi(x)$ is the digamma function. 

Then we set $\frac{dh(\raisedchi)}{dH} = 0$, and solve for $H$. Since $\psi(x)$ is a monotonically increasing function on $\mathbb{R^+}$, this implies that $\frac{dh(\raisedchi)}{dH}$ has one solution. Since $\frac{3}{2} - H = 2 - 2H$ only when $H = \half$, this implies that $h(\raisedchi)$ achieves a unique maximum at this point.
\end{proof}

This aligns with our intuition, that the highest uncertainty occurs for this model when it is uncorrelated and equal to white Gaussian noise, as it simplifies to $X_n = \epsilon_n$, identical to FGN processes. This explains why the maxima coincide for the two processes, given the same process variance, although ARFIMA(0,d,0) appears to higher differential entropy across the entire parameter range, when not at $H = 0.5$. This provides support that the maximum entropy process for LRD covariance constraints is the ARFIMA class, which echos previous results in this area such as Burg's Theorem (\cite{choi1984information}), that the AR and ARMA class of processes are the maximum entropy models given appropriate constraints on the covariances and impulse responses (see~\cite{ihara1984maximum, franke1985arma}).

We have shown in this section that the behaviour for the ARFIMA(0,d,0) model differs from that of FGN in the behaviour of their CSRD processes. This is a surprising discovery and warrants further investigation. Both models, however, have much less uncertainty as the strength of the positive correlations increases, as well as a maximum uncertainty occurring for uncorrelated processes. Hence, we may be able to characterise the behaviour of LRD processes on the entropy rate as tending to $-\infty$ as the strength of correlations increases. 

In remainder of the paper we look at other information theoretic measures as way to characterise the behaviour of SRD and LRD processes.

\section{Mutual Information and Excess Entropy for Long Range Dependent Processes}\label{mutual_information}

 In this section we continue analysing of the differential entropy rate for stochastic processes with power-law decaying covariance function. We investigate the links between the amount of entropy that is accumulated during the convergence of the conditional entropy to the entropy rate and the amount of information that is shared between the past and future of a stochastic process. 

 We extend the standard notion of mutual information to the special case of mutual information between past and future, $I_\pf$, which will measure the amount of information about the infinite future, given knowledge of the infinite past of stochastic processes.

\begin{definition}
	The mutual information for continuous random variables is defined as
	\begin{align*}
	 I(X;Y) = \int_{\mathcal{Y}} \int_{\mathcal{X}} f(x, y) \log\left(\frac{f(x,y)}{f(x)f(y)}\right)dx dy,
	 \end{align*} 
	 and in particular the mutual information between past and future with $n$ lags, $I^{(n)}$, for a stochastic process $\{X_i\}_{i \in \mathbb{Z}}$ is defined as $I(\{X_s, s < 0\},\{X_s, s \ge n\})$. The case with $n=0$ is called the mutual information between past and future, $I_\pf$, and is of special interest.
\end{definition}


We present a theorem from~\cite{Li_2004}, that links the value of $I_\pf$, and autocovariance function and the Fourier coefficients of the logarithm of the spectral density function.

\begin{theorem}[From~\cite{Li_2004}]\label{Li_thm_1}
	Let $\{X_i\}_{i \in \mathbb{Z}}$ be a stationary Gaussian stochastic process:
	\begin{itemize}
		\item if the spectral density $f(\lambda)$ is continuous and $f(\lambda) > 0$, then $I_\pf$ is finite if and only if the autocovariance function satisfies the condition $\sum_{k=-\infty}^{\infty} k \gamma(k)^2 < \infty$.
		\item $I_\pf$ is finite if and only if the cepstrum coefficients, $b_k = \frac{1}{2 \pi} \int_{-\pi}^{\pi} \log f(\lambda)e^{-ik\lambda} d\lambda$, satisfy the condition $\sum_{k=1}^{\infty} kb_k^2 < \infty$. In this case, $I_{p-f} = \frac{1}{2} \sum_{k=1}^{\infty} kb_k^2$.
	\end{itemize}
\end{theorem}
\begin{remark}
	The convergence of the sum $\sum_{k=-\infty}^{\infty} k b_k^2$ requires that $\sum_{k=1}^{\infty} k b_k^2 < \infty$ and $\sum_{k=1}^{\infty} -k b_{-k}^2 < \infty$ separately.
\end{remark}
	
\noindent This theorem gives us a way to classify whether processes have infinite mutual information between past and future, in this paper we will use this quantity to analyse convergence towards the entropic rate, in particular for LRD stochastic processes. 

In the next result we make an explicit link between LRD processes and the finiteness of the mutual information between past and future. This perspective provides us with a characterisation of LRD processes, these are processes that ``share infinite information from the past to the future".

\begin{theorem}\label{mutual_infinite}
	The Mutual Information between past and future,  $I_\pf$, for stationary LRD Gaussian processes is infinite.
\end{theorem}

\begin{proof}
To analyse LRD processes we use the spectral density asymptotic representation around the origin, $f(\lambda) \sim c_f|\lambda|^{1-2H}$, as we can see the divergence by considering the asymptotic behaviour around the singularity at the origin. Hence, as $b_k = \frac{1}{2 \pi} \int_{-\pi}^{\pi} \log f(\lambda) e^{-ik\lambda} d \lambda$, then
\begin{align*}
b_k &\sim \frac{1}{2 \pi} \int_{-\pi}^{\pi} \log\left(c_f|\lambda|^{1-2H}\right) e^{-ik\lambda} d \lambda,\\
&= \frac{\log c_f}{2 \pi} \int_{-\pi}^{\pi} e^{-ik\lambda} d \lambda + \frac{1-2H}{2 \pi} \int_{-\pi}^{\pi} \log |\lambda| e^{-ik\lambda} d \lambda.\\
\end{align*}
We split the integral into positive and negative components,
\begin{align*}
b_k &= \frac{2 \log c_f \sin\left(\pi k\right)}{2 \pi k} + \frac{1-2H}{2 \pi} \int_{0}^{\pi} \log \lambda \left(e^{-ik\lambda} + e^{ik\lambda}\right) d \lambda.
\end{align*}
Since $\sin(\pi k) = 0, \forall k \in \mathbb{Z}$, the first term vanishes and we only need to consider the integral. This can be decomposed into a trigonometric representation since $e^{-ik\lambda} + e^{ik\lambda} = 2\cos\left(k \lambda\right)$, so
\begin{align*}
b_k &\sim \frac{1-2H}{2 \pi} \int_{0}^{\pi} 2 \log \lambda \cos\left(k \lambda\right) d \lambda.
\end{align*}
Integrating this expression by parts, we get
\begin{align*}
	b_k &\sim \frac{1-2H}{2 \pi} \left[\frac{\log \lambda \sin\left(k\lambda\right)}{k}\right]_0^\pi - \int_{0}^{\pi} \frac{\sin\left(k\lambda\right)}{k\lambda} d \lambda.
\end{align*}
To analyse the integral in the second term we use the substitution, $u = k\lambda$, and therefore
\begin{align*}
	\int_{0}^{\pi} \frac{\sin\left(k\lambda\right)}{k\lambda} d \lambda &= \int_{0}^{k\pi} \frac{\sin\left(u\right)}{u} du.
\end{align*}
Then undoing the u-substitution and noting that Si$(x) = \int_{0}^{x} \frac{\sin t}{t} dt$, is the Sine integral,
\begin{align*}
 b_k &\sim \frac{1-2H}{\pi}\left[\frac{\log \lambda \sin\left(k \lambda\right) - \text{Si}(\lambda)}{k}\right]_0^\pi,\\
&= \frac{\left(1-2H\right)\left(-\text{Si}\left(\pi\right)\right)}{\pi k}.
\end{align*}

The partial sum has the asymptotic form,
\begin{align*}
\sum_{k=1}^{m} k b_k^2 &\sim \sum_{k=1}^{m} \frac{\left(1-2H\right)^2\ Si\left(\pi\right)^2}{\pi^2 k},\\
&= \frac{\left(1-2H\right)^2\ Si\left(\pi\right)^2}{\pi^2} \sum_{k=1}^{m} \frac{1}{k},\\
&\sim \frac{\left(1-2H\right)^2\ Si\left(\pi\right)^2}{\pi^2} \log m,
\end{align*}
because the harmonic series, $\sum_{k=1}^{m} \frac{1}{k} \sim \log(m)$ as $m \rightarrow \infty$. Hence, as $m \rightarrow \infty$, the sum, $\sum_{k=1}^{m} k b_k^2$, diverges. This implies that the sum, $\sum_{k=-\infty}^{\infty} k b_k^2$, diverges and therefore by Theorem~\ref{Li_thm_1}, $I_\pf$ is infinite.
\end{proof}

\begin{remark}
We can quite easily show this result with the additional assumptions that the spectral density, $f(\lambda)$ is positive and continuous. The asymptotic expression of the autocovariance function, $\gamma(k) \sim \sigma^2 c_{\rho} |k|^{-\alpha}$. Hence, considering the sum, $\sum_{k=1}^{\infty} k \gamma(k)^2$, from the first part of Theorem~\ref{Li_thm_1},
\begin{align*}
	\sum_{k=1}^{\infty} k \gamma(k)^2 &\sim \sum_{k=1}^{\infty} k (\sigma^2 c_{\rho} |k|^{-\alpha})^2 \\
	&= \sigma^4 c_{\rho}^2 \sum_{k=1}^{\infty} k^{1-2\alpha}.
\end{align*}
This sum diverges in the parameter range, $\alpha \in (0,1)$, and hence by Theorem~\ref{Li_thm_1}, $I_\pf$ is infinite for LRD processes.

There exist many processes that have infinite excess entropy but are not long range dependent. Some examples are given, including deterministic processes, in~\cite{crutchfield_feldman_2003}.
\end{remark}


\cite{crutchfield_feldman_2003} analysed a quantity named excess entropy, $\sum_{n=1}^{\infty} \left(H(X_n|X_{n-1}, \ldots, X_0) - H(\raisedchi)\right)$, for the Shannon entropy $H$ and corresponding entropy rate $H(\raisedchi)$, which has been shown to be equivalent to the mutual information between past and future. This has been named, with implicit interpretation, as stored information (\eg \cite{shaw1984dripping}), effective measure complexity (\eg \cite{grassberger1986toward, lindgren1988complexity}), predictive information (\eg \cite{nemenman2000information}). Importantly, it has been used to measure the convergence rate of the conditional entropy, based on past observations, to the entropy rate. We aim to extend this result to differential entropy, and then the question of classification of LRD processes via the amount of shared information can be made by the convergence rate to the entropy rate. We extend the definition of the excess entropy to the case of differential entropy.

\begin{definition}\label{def:differential_excess_entropy}
	The differential \emph{excess} entropy, $E$, of a stochastic process, $\{X_i\}_{i \in \mathbb{N}}$, is defined as,
	\begin{align*}
	E &= \sum_{n=1}^{\infty} \left(h_e(n) - h(\raisedchi)\right),\\
	&= \lim\limits_{n \rightarrow \infty} \left[ h(X_n, \ldots, X_0) - n h(\raisedchi)\right].
	\end{align*}
	where 
	\begin{align*}
		h_e(n) &= h(X_1, .. , X_n) - h(X_1, ... , X_{n-1}),\\
		 &= h(X_n | X_{n-1}, ... , X_1).
	\end{align*}
\end{definition}

We have the tools available to make an explicit link between the mutual information between past and future and the excess entropy of a continuous-valued, discrete-time stochastic process. This is an exact analogue of Proposition 8 from~\cite{crutchfield_feldman_2003} and has been stated utilising a different approach by~\cite{ding2016entropic}.

\begin{theorem}\label{mutual_excess}
	For a stationary, continuous-valued stochastic process, the mutual information between past and future, $I_\pf$, is equal to the differential excess entropy, $E$.
\end{theorem}

\begin{proof}
	The mutual information for a process $X$, with a past and future of $n$ observations,
	\begin{align*}
		I[\{X_s, -n \le s < 0\}; \{X_s, 0 \le s \le n\}] &= h(X_0, ..., X_{n-1}) - h(X_0, ... , X_{n-1} | X_{-n}, ... , X_{-1}),\\
		&= \sum_{i=0}^{n-1} h(X_i | X_{i-1}, ... , X_0) - \sum_{i=0}^{n-1} h(X_i | X_{i-1}, ... , X_{-n}),\\
		&= \sum_{i=0}^{n-1} \bigg(h(X_i | X_{i-1}, ... , X_0) - h(X_i | X_{i-1}, ... , X_{-n})\bigg),
	\end{align*}
	by the chain rule of differential entropy~\cite[pg. 253]{cover_thomas_2006}.
	Then we consider the mutual information between past and future, by taking the limit of the above expression as $n \rightarrow \infty$, which leads to
	\begin{align*}
		I_\pf &= \lim\limits_{n \rightarrow \infty} \left[ \sum_{i=0}^{n-1} \bigg(h(X_i | X_{i-1}, ... , X_0) - h(X_i | X_{i-1}, ... , X_{-n})\bigg)\right],\\
		&= \lim\limits_{n \rightarrow \infty} \left[ \sum_{i=0}^{\infty} \bigg( h(X_i | X_{i-1}, ... , X_0) - h(X_i | X_{i-1}, ... , X_{-n}) \bigg)\mathbbm{1}_{\{i \le n\}}\right].
	\end{align*}
	We define the sequence of measurable functions, $f_n(i)$ as
	\begin{align*}
		f_n(i) &= \bigg( h(X_i | X_{i-1}, ... , X_0) - h(X_i | X_{i-1}, ... , X_{-n}) \bigg)\mathbbm{1}_{\{i \le n\}},
	\end{align*}
	and we define the function, $g(i)$ as
	\begin{align*}
		g(i) &= h(X_i | X_{i-1}, ... , X_0) - h(\raisedchi).
	\end{align*}
	We want to show that $|f_n(i)| \le g(i)$ for all $n$ and for all $i \in \mathbb{N}$. In this case it is equivalent to showing that $f_n(i) \le g(i)$, since $f_n(i) \ge 0$ for all $n,i \in \mathbb{N}$, as the second term of $f_n(i)$ conditions on more random variables, and since conditioning cannot increase entropy this implies that $h(X_i | X_{i-1}, ... , X_0) \ge h(X_i | X_{i-1}, ... , X_{-n})$. We consider two cases, $i \le n$ and $i > n$, separately. In the case, $i > n$, we have that $f_n(i) = 0$, and since $g(i) \ge 0$ for all $i$, this implies that $f_n(i) \le g(i)$. Considering the second case, $i \le n$, we have that
	\begin{align*}
		f_n(i) &= \bigg(h(X_i | X_{i-1}, ... , X_0) - h(X_i | X_{i-1}, ... , X_{-n})\bigg),
	\end{align*}
	and therefore,
	\begin{align*}
		g(i) - f_n(i) &= h(X_i | X_{i-1}, ... , X_{-n}) - h(\raisedchi).
	\end{align*}
	Again, since conditioning does not increase entropy and the characterisation of entropy rate for stationary processes from Theorem~\ref{conditional_entropy_rate} this implies that $g(i) - f_n(i) \ge 0$ and therefore $g(i) \ge f_n(i)$ for all $n,i$ such that $i \le n$. Then we can apply the dominated convergence theorem~\cite[pg. 26]{Durrett2010}, since $f_n(i) \rightarrow (X_i | X_{i-1}, ... , X_0) - h(\raisedchi)$ pointwise, this implies that
	\begin{align*}
		I_\pf &= \lim\limits_{n \rightarrow \infty} \sum_{i=0}^{\infty} \bigg( h(X_i | X_{i-1}, ... , X_0) - h(X_i | X_{i-1}, ... , X_{-n}) \bigg)\mathbbm{1}_{\{i \le n\}},\\
		&=  \sum_{i=0}^{\infty} \lim\limits_{n \rightarrow \infty} \bigg( h(X_i | X_{i-1}, ... , X_0) - h(X_i | X_{i-1}, ... , X_{-n}) \bigg)\mathbbm{1}_{\{i \le n\}},\\
		&= \sum_{i=0}^{\infty} \bigg( h(X_i | X_{i-1}, ... , X_0) - h(\raisedchi) \bigg).
	\end{align*}
\end{proof}

\begin{remark}
	This proof is similar to that of Proposition 8 from~\cite{crutchfield_feldman_2003}. However, it is more rigorous since the limit is kept out the front of the sum while simultaneously applied to the second term in the sum. This approach using dominated convergence can resolve the issue in their proof.
\end{remark}

\noindent In~\cite{crutchfield_feldman_2003}, they analyse the excess entropy of discrete random variables to understand the convergence rate of the conditional entropy to the entropy rate. From this link we can utilise our knowledge of the mutual information between the past and future to classify the rate of convergence of the conditional entropy to the entropy rate. 


This result and Theorem~\ref{mutual_infinite}, lead to the following corollary which gives us an approach to understand the entropy rate convergence by conditional entropy, of Gaussian LRD processes and is a generalisation of a theorem stated for LRD FGN processes (\cite{ding2016entropic}).

\begin{corollary}
	The excess entropy of an LRD Gaussian process is infinite.
\end{corollary}
\begin{proof}
	This is shown by combining Theorem~\ref{mutual_infinite} and Theorem~\ref{mutual_excess}.
\end{proof}

\noindent We will use this idea in the subsequent sections to analyse the excess entropy, given its relationship to convergence to entropy rate which is noted for Shannon entropy rate by~\cite{crutchfield_feldman_2003}.

\section{Excess Entropy for Stationary Gaussian Processes}\label{excess_entropy}

In this section we investigate the behaviour of the excess entropy, $E$, for Gaussian processes which have an autocorrelation function which decays as a power law.

For stationary processes, the terms in the excess entropy will tend to zero as $n$ increases because
\begin{align*}
	\lim\limits_{n \rightarrow \infty} h_e(n) = \lim\limits_{n \rightarrow \infty} h(X_n | X_{n-1}, ... , X_1) = h(\raisedchi).
\end{align*}
However, we will investigate the nature of the convergence using the conditional entropy to gain some additional insight.

We begin by looking at the behaviour of the individual terms of the excess entropy series to understand why these terms decay to zero. 

Since $h(X_n ,... , X_1) = \frac{1}{2} \log\left((2 \pi e)^n |K^{(n)}| \right)$ for a finite collection of random variables of a Gaussian process~\cite[pg. 416]{cover_thomas_2006}, we have
\begin{align}
h_e(n) &= h(X_n ,... , X_1) - h(X_{n-1}, ... , X_1),\nonumber\\
&= \frac{1}{2} \log\left((2 \pi e)^n |K^{(n)}| \right) - \frac{1}{2} \log\left((2 \pi e)^{n-1} |K^{(n-1)}| \right),\nonumber\\
&= \frac{1}{2} \log\left(2 \pi e \frac{|K^{(n)}|}{|K^{(n-1)}|}\right),\label{eq:excess_entropy_term}
\end{align}
where $K^{(n)}$ is the autocovariance matrix of the process $X_n$. Note that $K^{(n)}$ is a Toeplitz matrix of size $n \times n$. 

We analyse each element of the infinite series of differential excess entropy by substituting the characterisation of $h_e(n)$ in~(\ref{eq:excess_entropy_term}) and utilising the entropy rate characterisation for Gaussian processes given in~(\ref{eqn:ihara}). This gives
\begin{align*}
h_e(n) - h(\raisedchi) &= \frac{1}{2} \log\left(2 \pi e \frac{|K^{(n)}|}{|K^{(n-1)}|}\right) - \frac{1}{2} \log (2 \pi e) - \frac{1}{4\pi} \int_{-\pi}^{\pi} \log f(\lambda) d\lambda,\\
&= \frac{1}{2} \log\frac{|K^{(n)}|}{|K^{(n-1)}|} - \frac{1}{4\pi} \int_{-\pi}^{\pi} \log f(\lambda) d\lambda. \numberthis \label{eqn:excess_ent_convergence}
\end{align*} 

\noindent The expression, (\ref{eqn:excess_ent_convergence}), informs us about the rate of convergence of the conditional entropy to the entropy rate, and about the general convergence of the excess entropy. If $\log f(\lambda)$ is integrable, there is a well known limit theorem by Szeg{\"o} (see~\cite{szego1915, bingham2012}), given below and used to evaluate the limit of $\frac{|K^{(n)}|}{|K^{(n-1)}|}$ as $n \rightarrow \infty$. 

First we show that $\log f(\lambda) \in L_1[-\pi, \pi]$, for LRD processes and hence that Szeg{\"o}'s limit theorem can be applied. We use the asymptotic form to analyse this case, as the issue with integrability is the singularity that exists at the origin of the spectral density functions. Therefore, the integrability of the asymptotic form implies the integrability for LRD processes.

\begin{theorem}
	For a spectral density, $f(\lambda)$ of a LRD stochastic process, $\log f(\lambda) \in L_1[-\pi, \pi]$.
\end{theorem}

\begin{proof}
	Using the asymptotic expression for spectral density for LRD processes,
	$f(\lambda) \sim c_f|\lambda|^{-\beta}, \text{  } \beta \in (0,1), \text{  } 0 < c_f < \infty$.\\
	Hence,
	\begin{align*}
	\int_{-\pi}^{\pi} \log(c_f|\lambda|^{-\beta}) d \lambda &= \int_{-\pi}^{\pi} \log c_f d \lambda - \beta \int_{-\pi}^{\pi} \log|\lambda| d \lambda,\\
	&= 2 \pi \log c_f - 2 \beta \int_{0}^{\pi} \log \lambda\ d\lambda,\\
	&= 2 \pi \log c_f - 2 \beta \left[\lambda \log \lambda - \lambda \right]_0^\pi,\\
	&= 2 \pi \log c_f - 2 \beta \left[\pi \log \pi - \pi \right],
	\end{align*}
	Since we have that $0\log 0 = 0$. By the finiteness of all of the terms, this implies that,
	\begin{align*}
	\int_{-\pi}^{\pi} \log f(\lambda) d \lambda < \infty.	
	\end{align*}
\end{proof}

The following theorem was originally formulated by~\cite{szego1915}, and then extended to include another term in~\cite{szego1952certain}. This theorem has a history of having the conditions generalised, and applications being found in many areas in functional analysis, statistics and probability~\cite[pg. 145-228]{grenander1958toeplitz}, to calculate functions of eigenvalues of Toeplitz matrices. The statement of the theorem that we will be using is by~\cite{bingham2012}, which is from the probabilistic perspective and includes the most recent generalisations in the conditions. 

\begin{theorem}[Szeg{\"o}'s Theorem (\cite{bingham2012})]\label{thm:szego_theorem}
	For a sequence of Toeplitz matrices, $\Gamma^{(n)} = [\gamma(|i-j|)]_{i,j}$, of increasing size $n$, where $\gamma(n) = \int_{-\pi}^{\pi} f(\lambda) e^{in\lambda} d \lambda$, and $f(\lambda)$ is a spectral density such that $\log f(\lambda) \in L_1[-\pi, \pi]$, then there exists a limit
	\begin{align*}
	\lim\limits_{n \to \infty} | \Gamma^{(n)} |^{\frac{1}{n}} = \exp \left\lbrace\frac{1}{2\pi}\int_{-\pi}^{\pi} \log f(\lambda) d\lambda \right\rbrace.
	\end{align*}
\end{theorem} 

\noindent We apply this theorem to the ratios of the autocovariance function, $K^{(n)}$, of the process, $\{X_n\}_{n \in \mathbb{N}}$. Due to the finiteness of this limit we can take the $n$th power, for a covariance matrix $K^{(n)}$. Which has the following form,
\begin{align*}
\lim\limits_{n \to \infty} | K^{(n)} | &= \exp \left\lbrace\frac{1}{2\pi}\int_{-\pi}^{\pi} \log f(\lambda) d\lambda \right\rbrace^n,
\end{align*}
taking the $n$th power. Applying this result to the limit of the $n$th term of the differential excess entropy (\ref{eqn:excess_ent_convergence}), gives
\begin{align*}
\lim\limits_{n \to \infty} h_e(n) - h(\raisedchi) 
&=  \frac{1}{2} \log \left( \exp \left\lbrace\frac{1}{2\pi}\int_{-\pi}^{\pi} \log f(\lambda) d\lambda \right\rbrace\right) - \frac{1}{4\pi} \int_{-\pi}^{\pi} \log f(\lambda) d\lambda,\\
&= 0.
\end{align*}

\noindent As $n \to \infty, h_e(n) - h(\raisedchi) \to 0$, which implies the convergence of the conditional entropy, conditioned on the infinite past, to the entropy rate, equivalent to the result of Theorem~\ref{conditional_entropy_rate}. We have gained some additional insight that the summands converge to zero in the limit as $n \rightarrow \infty$ due to the convegence of the ratio of subsequent covariance matrix determinants. However, this does not tell us anything about the rate of convergence of the terms or whether $E$ converges at all, but this approach can be extended to garner such information.

In following section we will use a stronger version of the Szeg{\"o} theorem, with an additional term in the limit. This gives an approach to analyse the convergence properties to the entropy rate that arise from observing the conditional entropy, conditioned on the past observations.



\section{Entropy rate convergence for LRD processes}

From the previous sections we have gained an understanding of some of the properties of the entropy rate function for common LRD models. In this section, we classify the convergence rate of the conditional entropy to the entropy rate for SRD and LRD processes.

In Section~\ref{excess_entropy} we used Szeg{\"o}'s theorem to show,
\begin{align*}
\lim\limits_{n \to \infty} \frac{|K^{(n)}|}{|K^{(n-1)}|} = \exp\left(\frac{1}{2\pi} \int_{-\pi}^{\pi} \log f(\lambda) d\lambda\right),
\end{align*}
if $\log f(\lambda) \in L_1[-\pi, \pi]$. This provides another perspective to explain why the conditional entropy converges to the entropy rate. The Strong Szeg{\"o} Theorem, provides an additional term to this limit with the same regularity conditions. The Szeg{\"o} theorem showed why the convergence occurred for the conditional entropy, and the strong Szeg{\"o} theorem we explain the convergence rate. We will state the version as given in~\cite{bingham2012}, which is more suited to probabilistic analysis.

For ease of notation we define the limit of the Szeg{\"o} theorem, 
\begin{align*}
G(\mu) := \exp\left( \frac{1}{2} \int_{-\pi}^{\pi} \log f(\lambda) d\lambda\right),
\end{align*}
and the partial autocorrelation coefficients, $\alpha_n$,  as,
\begin{align*}
\alpha_n = corr(X_n - P_{[1, n-1]}X_n, X_0 - P_{[1, n-1]}X_0),
\end{align*}
where $P_{[1, n-1]}$ is the projection onto the linear space spanned by $\{X_{-n}, ... , X_{-1}\}$ and the correlation function is defined as,
\begin{align*}
corr(X,Y) = \frac{\mathds{E}[X\bar{Y}]}{\sqrt{\mathds{E}[|X|^2]\mathds{E}[|Y|^2]}}, 
\end{align*}
for $X$ and $Y$ zero mean random variables. Note that $P_{[1, n-1]}X_0$ is the best linear predictor  of $X_0$ given the finite past of length $n$. 

The partial autocorrelation function is related to the autocorrelation function by the removal of the linear dependence on the variables within n lags. For example, in the case of finite lag processes, such as AR(p), the partial correlation function is 0 for a lag greater than p. In the case of the ARFIMA(p,d,q) process, the decay is slower for LRD parametrisations (\cite{inoue2002asymptotic}).

We define the Hardy space $\mathcal{H}^{\frac{1}{2}}$, which is the subspace of $\ell_2$ of sequences $a = (a_n)$ such that $||a||^2 := \sum_{n} (1 + |n|)|a_n|^2 < \infty$, which is a well defined norm on $\ell_2$. We use this to describe the conditions on the strong Szeg{\"o} theorem. Also, the cepstral coefficients that were defined in Theorem~\ref{Li_thm_1} will be used in the definition of the additional term of the Szeg{\"o} limit theorem, reinforcing their connection with the information theoretic perspective of stochastic processes.

\begin{theorem}[Strong Szeg{\"o} Theorem (\cite{bingham2012})]
	For a sequence of Toeplitz matrices, $\Gamma^{(n)} = [\gamma(|i-j|)]_{i,j}$, of increasing size $n$, for covariance function $\gamma(n)$, with associated spectral density $f(\lambda)$, such that $\log f(\lambda) \in L_1[-\pi, \pi]$, then 
	\begin{align*}
		\lim\limits_{n \to \infty} \frac{|\Gamma^{(n)}|}{G(\mu)^n} \to E(\mu),
	\end{align*}
	where,
	\begin{align*}
		E(\mu) := \prod_{j=1}^{\infty} (1 - |\alpha_j|^2)^{-j} = \exp\left(\sum_{k=1}^\infty k b_k^2\right),
	\end{align*}
	noting that these expressions may be infinite. The infinite product converges if and only if
	\begin{itemize}
		\item The strong Szeg{\"o} condition, $\alpha \in \mathcal{H}^{\frac{1}{2}}$ holds; or,
		\item The sum of cepstral coefficients $\sum_{k=1}^\infty k b_k^2 \in \mathcal{H}^{\frac{1}{2}}$.
	\end{itemize}
\end{theorem}

\noindent Using this theorem we will be able to characterise the convergence rate of the short range dependent processes that we have been considering in this paper.

The SRD and CSRD versions of ARFIMA and Fractional Gaussian Noise meet the conditions of the convergence of the infinite product and sum of Strong Szeg{\"o} Theorem. This is shown by Theorem~\ref{Li_thm_1}. For a positive, continuous spectral density  $\sum_{k=1}^{\infty} k \gamma(k)^2  < \infty$ if and only if  $I_{p-f} < \infty$, which in turn holds if and only if $\sum_{k=-\infty}^{\infty} |k||b_k|^2 < \infty$. In the proof of Theorem~\ref{mutual_infinite}, we showed that the boundary of the finiteness and infiniteness of $\sum_{k=1}^{\infty} kb_k^2$, coincided with the boundary between SRD and LRD processes. For example, for ARFIMA and Fractional Gaussian Noise, $\sum_{k=1}^{\infty} k \gamma(k)^2 < \infty$ when $H \le \frac{1}{2}$ and therefore by Theorem~\ref{mutual_infinite}, $\sum_{k=1}^{\infty} |k||b_k|^2 < \infty$.

One additional note about the conditions of the infinite product and sum in this theorem, the conditions for $\ell_1$ and $H < \frac{1}{2}$, coincide for sequences that decay as power laws, \ie $n^\alpha$, which is common when considering the convergence or divergence of sequences of LRD processes. However, it may be the case that a process may be in $\ell_1$ and not in $\mathcal{H}^\frac{1}{2}$, \emph{e.g.}, $a_n = \frac{1}{n\log n}$ (\cite{bingham2012}). 

Before we continue, we require a lemma to prove a theorem about the convergence rate of the differential entropy rate of SRD and LRD processes, which gives two different determinant limits for the results of the Szeg{\"o} limit theorems.

\begin{lemma}\label{determinant_ratio}
	For all discrete-time stationary Gaussian processes, 
	\begin{align*}
	\lim\limits_{n \to \infty} \left(|K^{(n)}| \right)^{\frac{1}{n}} &= \lim\limits_{n \to \infty} \frac{|K^{(n)}|}{|K^{(n-1)}|},\\
	\end{align*}
	where $K^{(n)}$ is the $n \times n$ autocovariance matrix of the process.
\end{lemma}

\begin{proof}
	For discrete-time stationary Gaussian processes, 
	\begin{align*}
	\lim\limits_{n \to \infty} \frac{1}{n} h(X_1, .. , X_n) &= \lim\limits_{n \to \infty} h(X_n|X_{n-1}, ... , X_1),
	\end{align*}
	by Theorem~\ref{conditional_entropy_rate} and the fact that both of the limits exist for stationary processes. The form of the joint entropy of a multivariate Gaussian random vector of length $n$~\cite[pg. 249]{cover_thomas_2006} is
	\begin{align*}
		h(X_1, .. , X_n) = \frac{1}{2} \log\left(2 \pi e |K^{(n)}|\right).
	\end{align*}
	Therefore as,
	\begin{align*}
		\lim\limits_{n \to \infty} \frac{1}{n} h(X_1, .. , X_n) &= \lim\limits_{n \to \infty} \frac{1}{n} \left(\frac{1}{2} \log\left((2 \pi e)^n |K^{(n)}|\right)\right),\\
		&= \lim\limits_{n \to \infty} \frac{1}{2} \log\left(2 \pi e |K^{(n)}|^{\frac{1}{n}}\right).
	\end{align*}
	Equating this for the entropy rate characterised using conditional entropy we get
	\begin{align*}
	\lim\limits_{n \to \infty} \frac{1}{2} \log\left(2 \pi e |K^{(n)}|^{\frac{1}{n}}\right)
	 &= \lim\limits_{n \to \infty} h(X_n, ... , X_1) - h(X_{n-1}, ... , X_1),\\
	&= \lim\limits_{n \to \infty} \frac{1}{2} \log\left((2 \pi e)^n |K^{(n)}|\right) - \frac{1}{2} \log\left((2 \pi e)^{n-1} |K^{(n-1)}|\right),\\ 
	&= \lim\limits_{n \to \infty} \frac{1}{2} \log\left((2 \pi e) \frac{|K^{(n)}|}{|K^{(n-1)}|}\right),
	\end{align*}
	Hence, we conclude that,
	\begin{align*}
	\lim\limits_{n \to \infty} (|K^{(n)}|)^{\frac{1}{n}} &= \lim\limits_{n \to \infty} \frac{|K^{(n)}|}{|K^{(n-1)}|}.
	\end{align*}
\end{proof}


\noindent Now we can characterise the convergence rate of the conditional entropy of SRD ARFIMA and Fractional Gaussian Noise processes to the differential entropy rate. Then we will show that a slower decay exists in the case of long range dependence.

\begin{theorem}\label{srd_convergence}
	For all discrete-time Gaussian SRD processes, such that the autocovariance function $\gamma \in \ell_1$, the convergence rate of the conditional entropy, $h(X_n|X_{n-1}, ..., X_1)$, to the differential entropy rate of the process is $O(n^{-1})$.
\end{theorem}

\begin{proof}
	
We rearrange the asymptotic expression of Szeg{\"o}'s strong theorem to get
\begin{align*}
\frac{|K^{(n)}|}{G(\mu)^n} &\sim E(\mu),\\
\implies |K^{(n)}|^{\frac{1}{n}} &\sim G(\mu)E(\mu)^{\frac{1}{n}}.
\end{align*}

Then we use the asymptotic form of the determinant limit from Lemma~\ref{determinant_ratio},
\begin{align*}
\frac{|K^{(n)}|}{|K^{(n-1)}|} &\sim G(\mu) \exp\left(\sum_{k=1}^\infty k b_k^2\right)^{\frac{1}{n}}, \\
\implies \frac{|K^{(n)}|}{|K^{(n-1)}|} &\sim G(\mu) \exp\left(\frac{1}{n} \sum_{k=1}^\infty k b_k^2\right).
\end{align*}

\noindent From the asymptotic form of the Szeg{\"o} theorem, Theorem~\ref{thm:szego_theorem}, the LHS will converge to $G(\mu)$ as $n$ increases. This implies that the convergence rate of the conditional entropy is controlled by the term, $\exp(\frac{1}{n} \sum_{k=1}^\infty k b_k^2)$, and its convergence rate to 1.  That is, $n \to \infty$
\begin{align*}
\left|G(\mu) \exp\left(\frac{1}{n} \sum_{k=1}^\infty k b_k^2\right) - G(\mu)\right| &= G(\mu)\left|\exp\left(\frac{1}{n} \sum_{k=1}^\infty k b_k^2\right) - \exp(0)\right|,\\
&= G(\mu).
\end{align*}
Since $\frac{1}{n} \sum_{k=1}^\infty k b_k^2 \to 0$, as $n \rightarrow \infty$.

Since $\sum_{k=1}^\infty k b_k^2$ is finite for $H \le \frac{1}{2}$, as $\gamma \in \ell_1$, this implies that the convergence is at the rate of $O(n^{-1})$ as
\begin{align*}
	\left|\exp\left(\frac{1}{n} \sum_{k=1}^\infty k b_k^2\right) - \exp(0)\right| &\to 0,\\
	\implies \left|\frac{1}{n} \sum_{k=1}^\infty k b_k^2 - 0\right| &\to 0,\\
	\text{and, } \left|\frac{1}{n} \sum_{k=1}^\infty k b_k^2\right| &\le \frac{C}{n}, C \in \mathbb{R^+}, \forall n \in \mathbb{Z^+}.
\end{align*}
\noindent For $C < \infty$, such that $C \ge \sum_{k=1}^\infty k b_k^2$.
\end{proof}

The key part of the proof that allows us to characterise the convergence rate of the conditional entropy is that the series, $\sum_{k=1}^\infty k b_k^2$, is finite. By Theorem~\ref{Li_thm_1}, we know that this is the boundary between LRD and SRD processes. This behaviour intuitively indicates that a limit would take a longer time to converge in the LRD case, similar to many estimators for LRD processes. Hence, we expect that the convergence to the entropy rate from the conditional entropy rate will be slower for LRD processes. This leads to the following theorem and characterisation of LRD processes through convergence of conditional entropy.




\begin{theorem}
	For all Gaussian LRD processes the convergence rate of the conditional entropy, $h(X_n|X_{n-1}, ..., X_1)$, to the differential entropy rate of the process is
	\begin{align*}
		O(\frac{\log(n^{(1-2H)^2})}{n}).
	\end{align*}
\end{theorem}

\begin{proof}
	Similar to the theorem above, we consider the convergence of the term, $\exp(\frac{1}{n} \sum_{k=1}^\infty k b_k^2)$, to 0, and use the following expansion of the term,
	\begin{align*}
		\lim\limits_{n \rightarrow \infty} \exp\left(\frac{1}{n} \sum_{k=1}^\infty k b_k^2\right) &= \lim\limits_{n \rightarrow \infty} \exp\left(\frac{1}{n} \lim\limits_{m \rightarrow \infty}\sum_{k=1}^m k b_k^2\right),\\
		&= \lim\limits_{n \rightarrow \infty} \lim\limits_{m \rightarrow \infty} \exp\left(\frac{1}{n} \sum_{k=1}^m k b_k^2\right).
	\end{align*}
	Where we use the continuity in the exponential function to exchange the limit in the function.
	From Theorem~\ref{mutual_infinite}, we showed that the rate of divergence of the partial sum of $kb_k^2$ is equal to,
	\begin{align*}
	\sum_{k=1}^{m} k b_k^2 &\sim \frac{\left(1-2H\right)^2\ Si\left(\pi\right)^2}{\pi^2} \log m,\\
	&= \frac{Si\left(\pi\right)^2}{\pi^2} \log \left(m^{\left(1-2H\right)^2}\right).
	\end{align*}
	Hence we can consider the limits as $n$ and $m$ tend to infinity,
	\begin{align*}
		\lim\limits_{n \rightarrow \infty} \lim\limits_{m \rightarrow \infty} \exp\left(\frac{1}{n} \sum_{k=1}^m k b_k^2\right) &\sim \lim\limits_{n \rightarrow \infty} \lim\limits_{m \rightarrow \infty} \exp\left(\frac{1}{n} \frac{Si\left(\pi\right)^2}{\pi^2} \log \left(m^{\left(1-2H\right)^2}\right) \right).
	\end{align*}
	If we take the $n \rightarrow \infty$ and $m \rightarrow \infty$ such that $m=n$, then we have 
	\begin{align*}
		\lim\limits_{n \rightarrow \infty} \exp\left(\frac{1}{n} \frac{Si\left(\pi\right)^2}{\pi^2} \log \left(n^{\left(1-2H\right)^2}\right) \right) \rightarrow 1.
	\end{align*}
	Hence,
	\begin{align*}
		\lim\limits_{n \rightarrow \infty} \frac{\log\left(n^{\left(1-2H\right)^2}\right)}{n} \rightarrow 0.
	\end{align*}	
	Which gives the rate of convergence of the conditional entropy to entropy rate.
\end{proof}

\noindent Therefore, the convergence of the conditional entropy to entropy rate is slower for LRD processes than SRD processes. This provides an information theoretic characterisation of LRD by convergence properties, similar to the covariance function, the sample mean and the parameters of linear predictors on the infinite past. The rate of convergence to the entropy rate decreases rapidly as $H \rightarrow 1$, because in $\left(1-2H\right)^2$ the influence of the Hurst parameter is squared. The convergence becomes much quicker as the Hurst parameter approaches 0.5, and eventually reaches the point where it converges immediately at $H=0.5$, as $h(X_i) = h(\raisedchi), \forall i \in \mathbb{Z^+}$ in the absence of any correlations.


\section{Conclusion}

In this paper, we are concerned with the behaviour of the differential entropy rate to understand and characterise the behaviour of LRD and SRD processes. Analysing two common LRD processes, FGN and ARFIMA(0,d,0), we have shown that the maximum occurs in the absence of correlations, \emph{i.e.}, $H=0.5$, and the differential entropy rate tends to the minimum, $-\infty$ as the strength of positive correlations increase, \emph{i.e.}, as we receive more information from correlations, the entropy of the process decreases. However, there is very different behaviour for negatively correlated processes, where ARFIMA(0,d,0) processes do not tend to $-\infty$ as the strength of the negative correlations increases. Further research is required to understand this behaviour for these processes.
 
In addition, we have made a link, similar to Shannon entropy, between the mutual information between past and future and excess entropy, meaning that the amount of shared information between the complete past of future of a process is the same as the additional information that accrues when converging to the entropy rate, based on past observations. This leads to a characterisation of LRD processes, as those having infinite mutual information between past and future. Using this and Szeg{\"o}'s limit theorems we then can classify LRD and SRD processes by their convergence rates, and show that LRD processes have slower convergence of the conditional entropy, conditioned on the past, to the entropy rate.

\section*{Acknowledgment}

This research was funded by CSIRO's Data61, the Australian Research Council's Centre of Excellence in Mathematical and Statistical Frontiers and Defence Science and Technology Group. Thanks to Adam Hamilton for his assistance in providing feedback on this manuscript.

\def\href#1#2{#2}
\bibliography{myBibliography} 



\end{document}